\documentclass[lettersize,journal]{IEEEtran}
\usepackage{amsmath,amsfonts}
\usepackage{algorithmic}
\usepackage{algorithm}
\usepackage{array}
\usepackage[caption=false,font=footnotesize,labelfont=rm,textfont=rm]{subfig}
\usepackage{textcomp}
\usepackage{stfloats}
\usepackage{url}
\usepackage{verbatim}
\usepackage{graphicx}
\usepackage{cite}

\usepackage{cases}
\usepackage{bm}
\newtheorem{theorem}{Theorem}
\newtheorem{proposition}{Proposition}
\newtheorem{proof}{Proof}

\begin{document}

\title{Edge Information Hub-Empowered 6G NTN: Latency-Oriented Resource Orchestration and Configuration}

\author{Yueshan Lin, Wei Feng, \IEEEmembership{Senior Member, IEEE}, Yunfei Chen, \IEEEmembership{Senior Member, IEEE}, Ning Ge, \IEEEmembership{Member, IEEE}, Zhiyong Feng, \IEEEmembership{Senior Member, IEEE}, and Yue Gao, \IEEEmembership{Fellow, IEEE}
	
	\thanks{Yueshan Lin, Wei Feng and Ning Ge are with the Department of Electronic Engineering, Tsinghua University, Beijing 100084, China (email: lin-ys17@tsinghua.org.cn, fengwei@tsinghua.edu.cn, gening@tsinghua.edu.cn).}
	\thanks{Yunfei Chen is with the Department of Engineering, University of Durham, Durham DH1 3LE, U.K. (e-mail: Yunfei.Chen@durham.ac.uk).}
	\thanks{Zhiyong Feng is with the School of Information and Communication Engineering, Beijing University of Posts and Telecommunications, Beijing 100084, China (email: ).}
	\thanks{Yue Gao is with the School of Computer Science, Fudan University, Shanghai 200433, China (email: yue.gao@ieee.org).}
}

\maketitle

\begin{abstract}
Quick response to disasters is crucial for saving lives and reducing loss. This requires low-latency uploading of situation information to the remote command center. Since terrestrial infrastructures are often damaged in disaster areas, non-terrestrial networks (NTNs) are preferable to provide network coverage, and mobile edge computing (MEC) could be integrated to improve the latency performance. Nevertheless, the communications and computing in MEC-enabled NTNs are strongly coupled, which complicates the system design. In this paper, an edge information hub (EIH) that incorporates communication, computing and storage capabilities is proposed to synergize communication and computing and enable systematic design. We first address the joint data scheduling and resource orchestration problem to minimize the latency for uploading sensing data. The problem is solved using an optimal resource orchestration algorithm. On that basis, we propose the principles for resource configuration of the EIH considering payload constraints on size, weight and energy supply. Simulation results demonstrate the superiority of our proposed scheme in reducing the overall upload latency, thus enabling quick emergency rescue.
\end{abstract}


\maketitle

\section{Introduction}
Every year, disasters (earthquakes, floods, fires, explosions, \textit{etc.}) cause significant losses to both human lives and economies. To reduce such losses, immediate emergency responses and rescue operations are necessary. This requires is real-time situation awareness of disaster areas \cite{Intro 10}. For this purpose, a large amount of data collected by field sensors will be uploaded to the remote command center with low latency. The data are then analyzed for decision making and rescue operations \cite{Intro 11}. Nevertheless, terrestrial communication infrastructures are often severely damaged in such disasters. Moreover, terrestrial networks might not cover areas where disasters take place, \textit{e.g.}, forests and oceans \cite{Intro 03}. Utilizing non-terrestrial networks (NTNs) such as satellites and unmanned aerial vehicles (UAVs) becomes crucial for disaster relief applications \cite{Intro 04} \cite{Intro 06}.

In practice, collecting sensing data through NTNs also face inherent limitations. Both satellites and UAVs are usually limited by communication resources \cite{Intro 05}. As the amount of data increases and the required latency decreases, a significant increase in the network throughput would be required, which may exceed the NTN's capability. Considering that the remote command center only requires key information on disaster areas to make rescue decisions, mobile edge computing (MEC) could be leveraged to process the sensing data and extract key information at the network edge \cite{Intro 07}. In this way, only the extracted key information needs to be uploaded through the satellite backhaul, the burden of which could be significantly reduced to reduce overall latency.

In MEC-enabled NTNs, the communications and computing are strongly coupled. In this case, orchestrating the multi-dimensional network resources separately might result in low resource efficiency and unsatisfactory latency performance. Therefore, an edge information hub (EIH) that incorporates communications, computing and storage capabilities is proposed to synergize heterogeneous parts and enable systematic design. In practice, the EIH could be deployed on UAVs to empower the NTN. There are two major problems in the EIH-empowered NTN. First, the data scheduling and network resource orchestration need to be jointly optimized for synergy, where the network resources include communications, computing and storage resources. Besides, since UAVs usaully have inherent limitations in size, weight and energy supply, it is necessary to derive the principles for resource configuration of the EIH. We thus focus on the optimal resource orchestration and the configuration for an EIH-Empowered NTN, which is envisioned a key part of the upcoming six-generation (6G) network. 

\IEEEpubidadjcol

\subsection{Related Works}
\hspace{3.2mm} \textit{1) Terrestrial Networks:} A number of terrestrial wireless solutions that support sensing data uploading as one of the Internet of Things (IoT) applications have been proposed. Some utilize the existing cellular networks for IoT applications, such as standards Extended Coverage GSM for IoT (EC-GSM-IoT) \cite{RW-terr 01} and Category-M LTE specifications \cite{RW-terr 02}. Others are based on low-power wide area network (LPWAN). The major technologies include Narrow-Band IoT (NB-IoT) \cite{RW-terr 03} and Long Range Radio (LoRa) \cite{RW-terr 04}.

\textit{2) Satellite Networks:} Since terrestrial networks are inherently limited by coverage, many existing studies consider utilizing satellite-enabled NTNs to provide services in remote and disaster areas. Sanctis \textit{et al.} \cite{RW-sat 01} introduced emergency management as well as two other application scenarios where satellite plays an important role. Specifically, satellites-enabled incident area networks could support both voice and data transmissions and wireless sensor and actuator communications in disaster areas. In \cite{RW-sat 02}, Centenaro \textit{et al.} presented a survey on deployment solutions for exploiting satellites to provide IoT services where terrestrial networks are unavailable. In the survey they discussed the pros and cons of satellite access as far as IoT traffic is concerned. Fang \textit{et al.} \cite{RW-sat 03} proposed three basic models of satellite-enabled networks to support ubiquitous IoT applications. For each basic model, a survey of the state-of-the-art technologies was provided, and future research directions were discussed.

\textit{3) Integrated Satellite-UAV Networks:} A number of existing studies further considered the integration of satellite and UAV networks for sensing data uploading. Bacco \textit{et al.} \cite{RW-satUAV 01} presented the design framework of a space information network (SIN) consisting of both satellites and UAVs to support IoT data exchanges. Examples of application scenarios as well as a possible relay solution were presented for the SIN. In \cite{RW-satUAV 02}, Zhu \textit{et al.} formulated a two-level queuing network to model the two-tier networks with UAV access and GEO satellite backhaul. In this study, closed-form expressions for the network backlog and delay bounds were derived, and the access scale was optimized to obtain optimal network performance. Ma \textit{et al.} \cite{RW-satUAV 03} jointly optimized the trajectory of UAVs, the bandwidth allocation among users, the transmit powers of UAVs and the selections of LEO satellites, to increase the average achievable rate and uploaded data amount while decreasing the consumed energy. In \cite{RW-satUAV 04}, Liu \textit{et al.} jointly optimized the subchannel allocation, the transmit power usage and the hovering times for data transmission efficiency maximization with a total latency constraint. The whole flight process of UAVs was considered for optimization and thus only the slowly-varying large-scale channel state information was used. Wang \textit{et al.} \cite{RW-satUAV 05} considered a integrated satellite-UAV framework where drones act as relays to upload the data from smart devices to low earth orbit satellites. They jointly optimized the smart devices connection scheduling, power control, and UAV trajectory to maximize the system capacity.

\textit{4) MEC-enabled Integrated Satellite-UAV Networks:} To further reduce the latency, MEC servers could be utilized to compress the data at the network edge by removing redundant information. MEC-empowered NTNs have been discussed in many studies. For instance, Lin \textit{et al.} \cite{RW-MEC 01} proposed three minimal integrating structures of MEC and NTNs, and established an on-demand network orchestration framework. In \cite{RW-MEC 02}, Kim \textit{et al.} investigated the data upload scheduling and path planning scheme for space-air-ground integrated edge computing systems, aimed at minimizing the total system energy cost. Ei \textit{et al.} \cite{RW-MEC 03} determined the optimal data upload scheduling and bandwidth allocation to minimize the total latency. In \cite{RW-MEC 04}, Chen \textit{et al.} jointly optimized the upload scheduling proportion and computing resource allocation to minimize the system energy consumption under time delay constraint. Chao \textit{et al.} \cite{RW-MEC 05} considered maximizing the profit of the MEC service provider by jointly designing the upload scheduling decisions and the UAV positions. In \cite{RW-MEC 06}, Waqar \textit{et al.} jointly optimized the upload scheduling decision, the bandwidth allocation, the computation resource allocation and the power usage of users, to minimize the weighted sum of total time delay and energy consumption. Ding \textit{et al.} \cite{RW-MEC 07} minimized the weighted sum energy consumption via joint user association, transmit precoding, task data assignment, and resource allocation. In \cite{RW-MEC 08}, Hu \textit{et al.} investigated the problem of joint optimization of the UAV 3D trajectory with resource allocation to maximizing the energy efficiency while satisfying users’ quality-of-experience. Chai \textit{et al.} \cite{RW-MEC 09} modelled the task data with dependencies as directed acyclic graphs and proposed a joint data upload scheduling and resource allocation scheme to improve the network efficiency. In \cite{RW-MEC 10}, Liu \textit{et al.} introduced a process-oriented framework that designs the whole process of data upload. Specifically, a latency minimization problem was formulated to optimize the data upload scheduling and power usage of users.

Although the studies mentioned above have made great progress in supporting sensing data uploading, there still exist research gaps. One is that the generalized optimal solution to the joint data scheduling and multi-dimensional resource orchestration problem has yet to be investigated. The existing studies discussed this problem by making assumptions to the system to simplify the problem and provide solutions for special cases. For instance, \cite{RW-MEC 03} proposed a joint data scheduling and subchannel allocation scheme assuming that the communication resources of the UAV-satellite link are equally allocated among users. In \cite{RW-MEC 10}, it was assumed that the communication and computation procedures are conducted sequentially instead of concurrently to simplify the problem. In this case, the communication and computing resources may not be fully utilized in the data uploading procedure. Another research gap is that existing studies mainly focus on the orchestration of communication and computation resources, but have not considered how much resources should be configured in the network \cite{RW-MEC 04} \cite{RW-MEC 06} \cite{RW-MEC 08}. It should be noted that solving the resource configuration problem relies on optimally orchestrating the multi-dimensional resources. Therefore, it is necessary to jointly consider the resource configuration problem and the resource orchestration problem.

\subsection{Main Contributions}
In this paper, we consider the systematic design of an EIH-empowered NTN to support low-latency sensing data upload in disaster relief scenarios. We investigate the joint data scheduling, communication resource allocation and computing capability orchestration problem. An optimal joint data scheduling and resource orchestration scheme is proposed to minimize the overall latency. Based on the optimal resource orchestration scheme, we derive several principles for the resource configuration of the EIH. Simulation results are presented to verify the conclusions and demonstrate the superiority of the proposed scheme. The main contributions of this paper are summarized as follows.

\begin{itemize}
	\item We model an EIH-empowered NTN to enable low-latency data uploading in disaster areas. In this model, the EIH is envisioned to synergize heterogeneous coupling parts, and enable systematic design. We formulate an overall upload latency minimization problem, where the data scheduling, the user-UAV transmission bandwidth allocation, the computing capability orchestration and the UAV-satellite data rate allocation are jointly considered.
	\item For the optimization problem, complicated piecewise functions exist in both the objective function and the constraints, rendering the problem hard to solve. We transform the problem by narrowing down its feasible region, and removing the piecewise functions. Accordingly, we equivalently recast the original problem into a convex form, and derive an optimal joint data scheduling and resource orchestration scheme.
	\item We derive the principles for resource configuration of the EIH, under its payload limitations in terms of size, weight and energy supply. Simulation results corroborate our theoretical achievements, and also demonstrate the superiority of our proposed scheme in reducing the overall upload latency.
\end{itemize}

The rest of this paper is organized as follows. In Section II, we introduce the system model of the EIH-empowered NTN. In Section III, we formulate the joint data scheduling and resource orchestration problem for latency minimization. We solve the problem and propose an optimal scheme in Section III. We further investigate the resource configuration problem and derive the principles for configuring the total computing capability in this section. Simulation results are presented in Section IV, while the conclusions of this paper are drawn in Section V.

\section{System Model}
As shown in Fig. 1, we consider an EIH-empowered NTN, which consists of $U$ users, a UAV equipped with an EIH, and a satellite to provide backhaul transmission. Each user $u$ has sensing data of size $D_u$ to upload. Fig. 2 presents the components of the EIH as well as the data flow diagram within it during the data upload process. The EIH incorporates a communication unit to the user, a communication unit to the satellite, a computing unit (\textit{i.e.}, an MEC server) and a storage unit. During the data upload process, the EIH receives user $u$'s uploaded data through its communication module and store the data in its storage module. For user $u$, the data could be divided into two parts, namely to-be-computed data (ratio $\eta_u$) and to-be-uploaded data (ratio $1-\eta_u$). The computing unit processes user $u$'s to-be-computed data with corresponding computing capability $F_u$ (cycles/s), and returns the outcome to the storage unit as to-be-uploaded data. The computing intensity of user $u$'s data is denoted as $\rho_u$ (cycles/bit). The computation could reduce the data size, where the computing output-to-input ratio is $\zeta_u$. The to-be-uploaded data are transmitted to the satellite from the EIH, with the transmission rate $R_u^S$. It should be noted that all transmissions and computations mentioned above could proceed concurrently.

\begin{figure}[t]
	\centering
	{\includegraphics[height=3in]{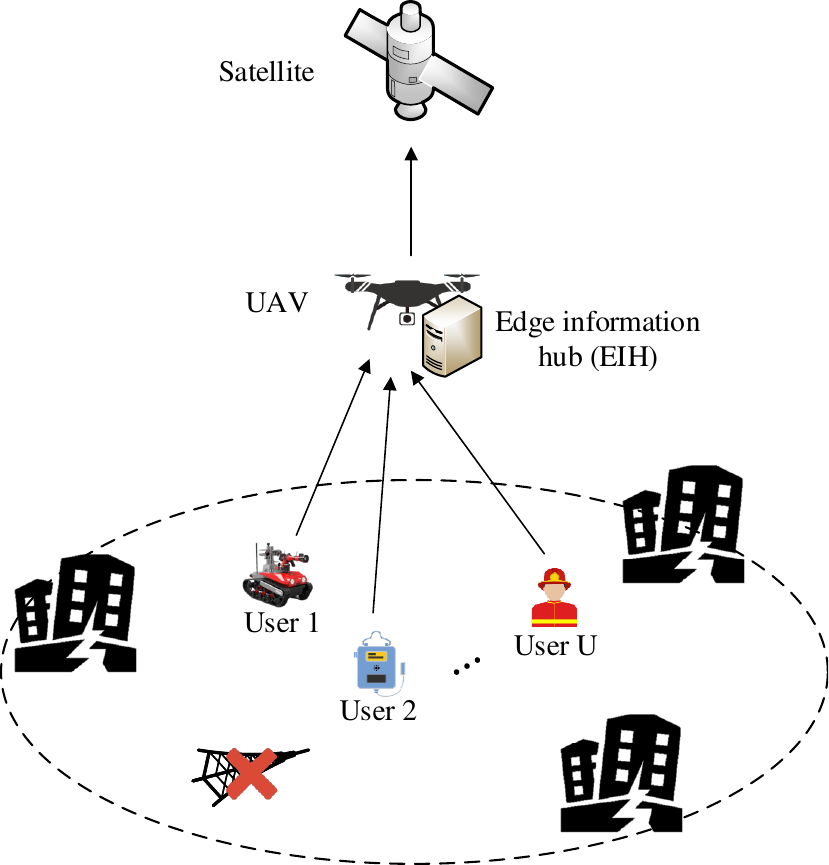}%
		\label{pic syst mod}}
	\caption{Illustration of the EIH-based non-terrestrial network.}
\end{figure}

\begin{figure}[b]
	\centering
	{\includegraphics[height=2.5in]{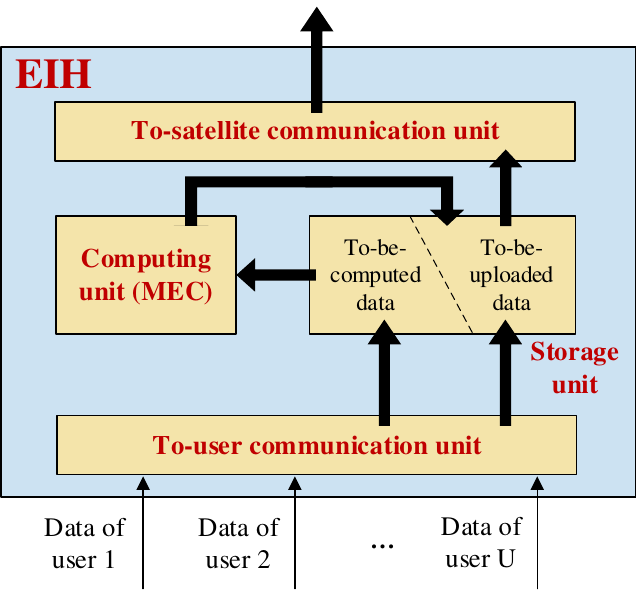}%
		\label{pic data flow}}
	\caption{Data flow diagram within the EIH during the sensing-data upload process.}
\end{figure}

The communication model of user-UAV transmission is presented as follows. We assume that both the users and the UAV have a single antenna, and therefore the received signal at the UAV from user $u$ is given by
\begin{equation}
	y_u = h_u x_u + n_u,
\end{equation}
where $x_u$ denotes the transmit signal, $n_u \sim \mathcal{CN}(0,\sigma^2)$ denotes the additive white Gaussian noise, and $h_u$ denotes the channel between the user and the UAV, which can be modeled as
\begin{equation}
	h_u = s \cdot l_u,
\end{equation}
where $s$ represents the fast-varying small-scale channel, which satisfies a complex Gaussian distribution with zero mean and unit variance (\textit{i.e.}, Rayleigh fading), and $l_u$ is the large-scale channel which can be expressed as follows \cite{syst 01}.
\begin{equation}
	l_u = 10^{ -\frac{1}{20} \left( \frac{A_0}{1+a\exp(-b(\theta_u-a))}+B_0 \right) },
\end{equation}
where $A_0 = \eta_{\rm LoS}-\eta_{\rm NLoS}$ and $B_0 = \eta_{\rm NLoS}+20\log_{10}(4\pi f d_u/c)$. In this model, $\eta_{\rm LoS}$, $\eta_{\rm NLoS}$, $a$ and $b$ are constants related to the propagation environment, $d_u$ and $\theta_u$ are the distance and elevation angle between user $u$ and the UAV respectively, $f$ denotes the carrier frequency and $c$ is the light speed. We assume that the users adopt a frequency division multiple access (FDMA) scheme to avoid interference. We use $r_u$ to denote the ergodic spectrum efficiency, and the ergodic transmission rate can therfore be formulated as
\begin{equation}
	R_u = B_u {\rm E}\left[ \log_2\left(1+\frac{p_u|h_u|^2}{\sigma^2}\right) \right] = B_u r_u,
\end{equation}
where $B_u$ is the allocated bandwidth for user $u$, $p_u = {\rm E}(|x_u|^2)$ denotes the transmission power, and $\sigma^2$ denotes the noise power. We introduce a parameter $r_u$ to denote the ergodic spectrum efficiency.

\begin{figure*}[t]
	\begin{subnumcases}{\label{expr_T} T_u(B_u,R_u^S,F_u,\eta_u) =}
		\vspace{1.5mm}
		\frac{D_u}{R_u^S}(\zeta_u\eta_u+1-\eta_u), & $\eta_u B_u r_u \geq \frac{F_u}{\rho_u}$, $\frac{\zeta_u F_u}{\rho_u}+(1-\eta_u) B_u r_u \geq R_u^S$, $\frac{F_u}{\rho_u} \geq \frac{\eta_u R_u^S}{\zeta_u\eta_u+1-\eta_u}$ \\
		\vspace{1.5mm}
		\frac{\eta_u D_u \rho_u}{F_u}, & $\eta_u B_u r_u \geq \frac{F_u}{\rho_u}$, $\frac{\zeta_u F_u}{\rho_u}+(1-\eta_u) B_u r_u \geq R_u^S$, $\frac{F_u}{\rho_u} < \frac{\eta_u R_u^S}{\zeta_u\eta_u+1-\eta_u}$ \\
		\vspace{1.5mm}
		\frac{\eta_u D_u \rho_u}{F_u}, & $\eta_u B_u r_u \geq \frac{F_u}{\rho_u}$, $\frac{\zeta_u F_u}{\rho_u}+(1-\eta_u) B_u r_u < R_u^S$ \\
		\vspace{1.5mm}
		\frac{D_u}{R_u^S}(\zeta_u\eta_u+1-\eta_u), & $\eta_u B_u r_u < \frac{F_u}{\rho_u}$, $(\zeta_u\eta_u+1-\eta_u) B_u r_u \geq R_u^S$ \\
		\vspace{1.5mm}
		\frac{D_u}{B_u r_u}, & $\eta_u B_u r_u < \frac{F_u}{\rho_u}$, $(\zeta_u\eta_u+1-\eta_u) B_u r_u < R_u^S$
	\end{subnumcases}
\end{figure*}

\begin{figure*}[t]
	\begin{subnumcases}{\label{expr_V} V_u(B_u,R_u^S,F_u,\eta_u) =}
		\vspace{1.5mm}
		\nonumber \frac{D_u}{B_u r_u}\left[B_u r_u - R_u^S - (1-\zeta_u) \frac{F_u}{\rho_u}\right], & $\eta_u B_u r_u \geq \frac{F_u}{\rho_u}$, $\frac{\zeta_u F_u}{\rho_u}+(1-\eta_u) B_u r_u \geq R_u^S$, \\
		\vspace{1.5mm}
		& \hspace{47mm} $\frac{F_u}{\rho_u} \geq \frac{\eta_u R_u^S}{\zeta_u\eta_u+1-\eta_u}$ \\
		\vspace{1.5mm}
		\nonumber \frac{D_u}{B_u r_u}\left[B_u r_u - R_u^S - (1-\zeta_u) \frac{F_u}{\rho_u}\right], & $\eta_u B_u r_u \geq \frac{F_u}{\rho_u}$, $\frac{\zeta_u F_u}{\rho_u}+(1-\eta_u) B_u r_u \geq R_u^S$, \\
		\vspace{1.5mm}
		& \hspace{47mm} $\frac{F_u}{\rho_u} < \frac{\eta_u R_u^S}{\zeta_u\eta_u+1-\eta_u}$ \\
		\vspace{1.5mm}
		\frac{D_u}{B_u r_u}\left(\eta_u B_u r_u - \frac{F_u}{\rho_u}\right), & $\eta_u B_u r_u \geq \frac{F_u}{\rho_u}$, $\frac{\zeta_u F_u}{\rho_u}+(1-\eta_u) B_u r_u < R_u^S$ \\
		\vspace{1.5mm}
		\frac{D_u}{B_u r_u}[(\zeta_u\eta_u+1-\eta_u) B_u r_u - R_u^S], & $\eta_u B_u r_u < \frac{F_u}{\rho_u}$, $(\zeta_u\eta_u+1-\eta_u) B_u r_u \geq R_u^S$ \\
		\vspace{1.5mm}
		0, & $\eta_u B_u r_u < \frac{F_u}{\rho_u}$, $(\zeta_u\eta_u+1-\eta_u) B_u r_u < R_u^S$
	\end{subnumcases}
	\hrulefill
\end{figure*}

We use $T_u(B_u,R_u^S,F_u,\eta_u)$ to denote the total communication and computing latency for user $u$ to upload its sensing data, and $V_u(B_u,R_u^S,F_u,\eta_u)$ to denote the minimum storage required by user $u$. Their specific expressions are shown in \textit{Proposition 1}.

\begin{proposition}
	The expression of $T_u(B_u,R_u^S,F_u,\eta_u)$ is given by (\ref{expr_T}), and the expression of $V_u(B_u,R_u^S,F_u,\eta_u)$ is given by (\ref{expr_V}) in the next page.
	\label{proposition 1}
\end{proposition}
\begin{proof}
	See Appendix A.
\end{proof}

\section{Problem Formulation and Proposed Scheme}
In this section, we formulate the joint data scheduling and resource orchestration problem for latency minimization and propose our solution to the problem. Specifically, we first obtain the optimal data scheduling variable. Then we transform the optimization problem by narrowing its feasible region, so that the complicated piecewise functions are simplified for a viable solution. Based on this solution, we further investigate the resource configuration problem and derive several principles for configuring the total computing capability.

\subsection{Joint Data Scheduling and Resource Orchestration}
We minimize the overall data upload latency by optimally determining each user's allocated bandwidth in the user-UAV transmission $B_u$, allocated data rate in the UAV-satellite transmission $R_u^S$, allocated computing resource $F_u$, as well as data scheduling $\eta_u$. The optimization problem can be formulated as
\begin{subequations} \label{problem 1}
	\begin{align}
		\min_{\mathbf{B}, \mathbf{R}^S, \atop \mathbf{F}, \bm{\eta}} & \max_{u} \  T_u(B_u,R_u^S,F_u,\eta_u) \label{problem 1a} \\
		\mathrm{s.t.} \ \ & \sum_{u=1}^{U} V_u(B_u,R_u^S,F_u,\eta_u) \leq V_{\mathrm{total}}, \label{problem 1b} \\
		& \sum_{u=1}^{U} B_u \leq B_{\mathrm{total}}, \label{problem 1c} \\
		& \sum_{u=1}^{U} R_u^S \leq R^S_{\mathrm{total}}, \label{problem 1d} \\
		& \sum_{u=1}^{U} F_u \leq F_{\mathrm{total}}, \label{problem 1e} \\
		& B_u \geq 0, \  \forall u = 1,...,U, \label{problem 1f} \\
		& R_u^S \geq 0, \  \forall u = 1,...,U, \label{problem 1g} \\
		& F_u \geq 0, \  \forall u = 1,...,U, \label{problem 1h} \\
		& 0 \leq \eta_u \leq 1, \  \forall u = 1,...,U, \label{problem 1i}
	\end{align}
\end{subequations}
where $V_{\mathrm{total}}$, $B_{\mathrm{total}}$, $R^S_{\mathrm{total}}$, and $F_{\mathrm{total}}$ denote the total storage on UAV, the total bandwidth of user-UAV transmission, the total data rate of UAV-satellite transmission and the total computing capability of the computing unit, respectively. As shown in (\ref{expr_T}) and (\ref{expr_V}), both $T_u(B_u,R_u^S,F_u,\eta_u)$ and $V_u(B_u,R_u^S,F_u,\eta_u)$ are complicated piecewise functions, and thus problem (\ref{problem 1}) is difficult to solve.

\begin{table*}[t]
	\caption{Optimal data scheduling variable and corresponding optimized function values.}
	\centering
	\renewcommand{\arraystretch}{2.5}
	\begin{tabular}{|m{2.3in}|m{1.05in}|m{1.05in}|m{1.9in}|}
		\hline
		\textbf{Condition} & $\eta_u^{\mathrm{opt}}$ & $T_u^{\eta-\mathrm{opt}}$ & $V_u^{\eta-\mathrm{opt}}$ \\ \hline
		$B_u r_u < R_u^S$ & $0$ & $\frac{D_u}{B_u r_u}$ & $0$ \\ \hline
		$R_u^S \leq B_u r_u < \frac{R_u^S}{\zeta_u}$, $\frac{F_u}{\rho_u} < \frac{B_u r_u - R_u^S}{1-\zeta_u}$  & $\frac{F_u}{F_u(1-\zeta_u)+\rho_u R_u^S}$ & $\frac{D_u \rho_u}{F_u(1-\zeta_u)+\rho_u R_u^S}$ & $\frac{D_u}{B_u r_u}\left[B_u r_u - R_u^S - (1-\zeta_u) \frac{F_u}{\rho_u}\right]$ \\ \hline
		$R_u^S \leq B_u r_u < \frac{R_u^S}{\zeta_u}$, $\frac{F_u}{\rho_u} \geq \frac{B_u r_u - R_u^S}{1-\zeta_u}$ & $\frac{B_u r_u - R_u^S}{(1-\zeta_u) B_u r_u}$ & $\frac{D_u}{B_u r_u}$ & $0$ \\ \hline
		$B_u r_u \geq \frac{R_u^S}{\zeta_u}$, $\frac{F_u}{\rho_u} < \frac{R_u^S}{\zeta_u}$ & $\frac{F_u}{F_u(1-\zeta_u)+\rho_u R_u^S}$ & $\frac{D_u \rho_u}{F_u(1-\zeta_u)+\rho_u R_u^S}$ & $\frac{D_u}{B_u r_u}\left[B_u r_u - R_u^S - (1-\zeta_u) \frac{F_u}{\rho_u}\right]$ \\ \hline
		$B_u r_u \geq \frac{R_u^S}{\zeta_u}$, $\frac{R_u^S}{\zeta_u} \leq \frac{F_u}{\rho_u} < B_u r_u$ & $1$ & $\frac{\zeta_u D_u}{R_u^S}$ & $\frac{D_u}{B_u r_u}\left[B_u r_u - R_u^S - (1-\zeta_u) \frac{F_u}{\rho_u}\right]$ \\ \hline
		$B_u r_u \geq \frac{R_u^S}{\zeta_u}$, $\frac{F_u}{\rho_u} \geq B_u r_u$ & $1$ & $\frac{\zeta_u D_u}{R_u^S}$ & $\frac{D_u}{B_u r_u}(\zeta_u B_u r_u - R_u^S)$ \\ \hline
	\end{tabular}
	\label{eta-opt}
\end{table*}

\subsection{Optimal Data Scheduling}
By observing the expressions of $T_u(B_u,R_u^S,F_u,\eta_u)$ and $V_u(B_u,R_u^S,F_u,\eta_u)$ in (\ref{expr_T}) and (\ref{expr_V}), we present the following theorem
\begin{theorem}
	Assume that the values of variables $B_u$, $R_u^S$ and $F_u$ are given, an optimal value of variable $\eta_u^{\mathrm{opt}}$ can be determined to minimize both functions $T_u(B_u,R_u^S,F_u,\eta_u)$ and $V_u(B_u,R_u^S,F_u,\eta_u)$. The expressions of $\eta_u^{\mathrm{opt}}$ and the corresponding minimized function values $T_u^{\eta-\mathrm{opt}}(B_u,R_u^S,F_u) = T_u(B_u,R_u^S,F_u,\eta_u^{\mathrm{opt}})$ and $V_u^{\eta-\mathrm{opt}}(B_u,R_u^S,F_u) = V_u(B_u,R_u^S,F_u,\eta_u^{\mathrm{opt}})$ are given in Table \ref{eta-opt}.
	\label{theorem eta-opt}
\end{theorem}
\begin{proof}
	See Appendix B.
\end{proof}

Using Theorem \ref{theorem eta-opt}, we present a new optimization problem (\ref{problem 1}) as 
\begin{subequations} \label{problem 2}
	\begin{align}
		\min_{\mathbf{B}, \mathbf{R}^S, \mathbf{F}} & \max_{u} \  T_u^{\eta-\mathrm{opt}}(B_u,R_u^S,F_u) \label{problem 2a} \\
		\mathrm{s.t.} \ \ & \sum_{u=1}^{U} V_u^{\eta-\mathrm{opt}}(B_u,R_u^S,F_u) \leq V_{\mathrm{total}}, \label{problem 2b} \\
		& \sum_{u=1}^{U} B_u \leq B_{\mathrm{total}}, \label{problem 2c} \\
		& \sum_{u=1}^{U} R_u^S \leq R^S_{\mathrm{total}}, \label{problem 2d} \\
		& \sum_{u=1}^{U} F_u \leq F_{\mathrm{total}}, \label{problem 2e} \\
		& B_u \geq 0, \  \forall u = 1,...,U, \label{problem 2f} \\
		& R_u^S \geq 0, \  \forall u = 1,...,U, \label{problem 2g} \\
		& F_u \geq 0, \  \forall u = 1,...,U. \label{problem 2h}
	\end{align}
\end{subequations}

The relationship between problem (\ref{problem 1}) and problem (\ref{problem 2}) is illustrated in the following theorem.
\begin{theorem}
	The optimal objective functions in optimization problem (\ref{problem 1}) and optimization problem (\ref{problem 2}) are equal.
\end{theorem}
\begin{proof}
	We adopt the method of proof by contradiction. Assume that the theorem does not hold. Since problem (\ref{problem 2}) has the same objective function as problem (\ref{problem 1}) but has a smaller feasible region, this is equivalent to the following statement. For problem (\ref{problem 1}) there exists at least one set of optimal variables $(\mathbf{B}_{0}, \mathbf{R}^S_{0}, \mathbf{F}_{0}, \bm{\eta}_{0})$ lying outside the feasible region of problem (\ref{problem 2}), and satisfying that $\underset{u}{\max} \  T_u(B_{0,u},R_{0,u}^S,F_{0,u},\eta_{0,u}) < \underset{u}{\max} \  T_u^{\eta-\mathrm{opt}}(B_{u},R_{u}^S,F_{u})$ holds for all feasible solutions $(\mathbf{B}, \mathbf{R}^S, \mathbf{F})$ of problem (\ref{problem 2}).
	
	Consider $(\mathbf{B}_{0}, \mathbf{R}^S_{0}, \mathbf{F}_{0})$ as a solution for problem (\ref{problem 2}). Obviously constraints (\ref{problem 2c})-(\ref{problem 2h}) hold. For each $(B_{0,u},R_{0,u}^S,F_{0,u})$ we obtain the corresponding $\eta_u^{\mathrm{opt}}$ according to Table \ref{eta-opt}. Therefore, we have
	\begin{align}
		\nonumber V_u^{\eta-\mathrm{opt}}(B_{0,u},R_{0,u}^S,F_{0,u}) & = V_u(B_{0,u},R_{0,u}^S,F_{0,u},\eta_u^{\mathrm{opt}}) \\
		& \leq V_u(B_{0,u},R_{0,u}^S,F_{0,u},\eta_{0,u})
	\end{align}
	This means that constraint (\ref{problem 2b}) also holds, and $(\mathbf{B}_{0}, \mathbf{R}^S_{0}, \mathbf{F}_{0})$ is a feasible solution of problem (\ref{problem 2}). Besides, we have
	\begin{align}
		\nonumber T_u^{\eta-\mathrm{opt}}(B_{0,u},R_{0,u}^S,F_{0,u} & ) = T_u(B_{0,u},R_{0,u}^S,F_{0,u},\eta_u^{\mathrm{opt}}) \\
		& \leq T_u(B_{0,u},R_{0,u}^S,F_{0,u},\eta_{0,u})
	\end{align}
	and thus the inequality $\underset{u}{\max} \  T_u(B_{0,u},R_{0,u}^S,F_{0,u},\eta_{0,u}) \geq \underset{u}{\max} \  T_u^{\eta-\mathrm{opt}}(B_{0,u},R_{0,u}^S,F_{0,u})$ holds, which causes contradiction. Therefore the assumption does not hold and the theorem is proved.
\end{proof}

\subsection{Optimal Resource Orchestration}
Although we transform the problem equivalently, functions $T_u^{\eta-\mathrm{opt}}(B_u,R_u^S,F_u)$ and $V_u^{\eta-\mathrm{opt}}(B_u,R_u^S,F_u)$ are still complicated piecewise functions, which makes the problem hard to solve.

Based on problem (\ref{problem 2}), we present a new optimization problem as follows which further narrows down the feasible region so that both piecewise functions are simplified.
\begin{subequations} \label{problem 3}
	\begin{align}
		\min_{\mathbf{B}, \mathbf{R}^S, \mathbf{F}} & \max_{u} \  \frac{D_u}{B_u r_u} \label{problem 3a} \\
		\mathrm{s.t.} \ \ & \sum_{u=1}^{U} B_u \leq B_{\mathrm{total}}, \label{problem 3b} \\
		& \sum_{u=1}^{U} R_u^S \leq R^S_{\mathrm{total}}, \label{problem 3c} \\
		& \sum_{u=1}^{U} F_u \leq F_{\mathrm{total}}, \label{problem 3d} \\
		& R_u^S \leq B_u r_u \leq \frac{R_u^S}{\zeta_u}, \  \forall u = 1,...,U, \label{problem 3e} \\
		& \frac{F_u}{\rho_u} = \frac{B_u r_u - R_u^S}{1-\zeta_u}, \  \forall u = 1,...,U, \label{problem 3f} \\
		& B_u \geq 0, \  \forall u = 1,...,U, \label{problem 3g} \\
		& R_u^S \geq 0, \  \forall u = 1,...,U, \label{problem 3h} \\
		& F_u \geq 0, \  \forall u = 1,...,U. \label{problem 3i}
	\end{align}
\end{subequations}
By introducing new constraints (\ref{problem 3e}) and (\ref{problem 3f}), according to the third row of Table \ref{eta-opt}, both piecewise functions have a certain range, with expressions $T_u^{\eta-\mathrm{opt}}(B_u,R_u^S,F_u) = \frac{D_u}{B_u r_u}$ and $V_u^{\eta-\mathrm{opt}}(B_u,R_u^S,F_u) = 0$. Therefore the objective function (\ref{problem 2a}) changes into (\ref{problem 3a}), and constraint (\ref{problem 2b}) is removed in problem (\ref{problem 3}) since it always holds. The relationship between problem (\ref{problem 2}) and problem (\ref{problem 3}) is illustrated in the following theorem.
\begin{theorem}
	The optimal objective functions in optimization problem (\ref{problem 2}) and optimization problem (\ref{problem 3}) are equal.
\end{theorem}
\begin{proof}
	See Appendix C.
\end{proof}

Finally, we introduce a slack variable $T$ and equivalently recast the problem as
\begin{subequations} \label{problem 4}
	\begin{align}
		\min_{\mathbf{B}, \mathbf{R}^S, T} \  & T \label{problem 4a} \\
		\mathrm{s.t.} \ \ & T \geq \frac{D_u}{B_u r_u}, \  \forall u = 1,...,U, \label{problem 4b} \\
		& \sum_{u=1}^{U} B_u \leq B_{\mathrm{total}}, \label{problem 4c} \\
		& \sum_{u=1}^{U} R_u^S \leq R^S_{\mathrm{total}}, \label{problem 4d} \\
		& \sum_{u=1}^{U} \rho_u \frac{B_u r_u - R_u^S}{1-\zeta_u} \leq F_{\mathrm{total}}, \label{problem 4e} \\
		& R_u^S \leq B_u r_u \leq \frac{R_u^S}{\zeta_u}, \  \forall u = 1,...,U, \label{problem 4f} \\
		& T \geq 0, \label{problem 4g} \\
		& B_u \geq 0, \  \forall u = 1,...,U, \label{problem 4h} \\
		& R_u^S \geq 0, \  \forall u = 1,...,U, \label{problem 4i}
	\end{align}
\end{subequations}
It can be verified that problem (\ref{problem 4}) is a convex optimization problem and can be solved directly. We denote the optimal solution to this problem as $\overline{T}, \overline{\mathbf{B}} = \left[\overline{B_1},...,\overline{B_U}\right], \overline{\mathbf{R}^S} = \left[\overline{R_1^S},...,\overline{R_U^S}\right]$, where $\overline{T}$ is the minimal overall data upload latency, and $\overline{\mathbf{B}}$ and $\overline{\mathbf{R}^S}$ are the optimal bandwidth allocation of user-UAV transmission and the optimal data rate allocation of UAV-satellite transmission, respectively. Moreover, the optimal computing capability orchestration $\overline{\mathbf{F}} = \left[\overline{F_1},...,\overline{F_U}\right]$ can be given by
\begin{equation}
	\overline{F_u} = \rho_u \frac{\overline{B_u} r_u - \overline{R_u^S}}{1-\zeta_u}, \  \forall u = 1,...,U,
\end{equation}
and the optimal data scheduling variable could be obtained by substituting $\overline{\mathbf{B}}, \overline{\mathbf{R}^S}, \overline{\mathbf{F}}$ to calculate $\bm{\eta}^{\rm opt}$ according to Table 1. By now, we have proposed the optimal joint data scheduling and resource orchestration scheme.

\subsection{Resource Configuration}
We further investigate the resource configuration problem based on the proposed data scheduling and resource allocation scheme. Specifically, we aim to obtain the relationship between the overall latency and the total configured resources $B_{\mathrm{total}}$, $R_{\mathrm{total}}^S$ and $F_{\mathrm{total}}$. We present the following theorems
\begin{theorem}
	Assume two optimization problems, both adopting the form of problem (\ref{problem 4}). For both problems, the corresponding parameters take the same values except for the total computing capability, where in the first problem it is $F_{\rm total}^{(1)}$ while in the second it is $F_{\rm total}^{(2)}$. If the following inequality holds
	\begin{equation}
		F_{\rm total}^{(1)} \leq F_{\rm total}^{(2)},
	\end{equation}
	then the optimal values of the objective functions of both problems, denoted as $\overline{T}^{(1)}$ and $\overline{T}^{(2)}$ respectively, satisfy
	\begin{equation}
		\overline{T}^{(1)} \geq \overline{T}^{(2)}.
	\end{equation}
\end{theorem}
\begin{proof}
	Denote the optimal solution to the first and second problem as $(\overline{\mathbf{B}}^{(1)}, \overline{\mathbf{R}^S}^{(1)}, \overline{T}^{(1)})$ and $(\overline{\mathbf{B}}^{(2)}, \overline{\mathbf{R}^S}^{(2)}, \overline{T}^{(2)})$, respectively. Since $F_{\rm total}^{(1)} \leq F_{\rm total}^{(2)}$, $(\overline{\mathbf{B}}^{(1)}, \overline{\mathbf{R}^S}^{(1)}, \overline{T}^{(1)})$ is also a feasible solution for the second problem, and because $(\overline{\mathbf{B}}^{(2)}, \overline{\mathbf{R}^S}^{(2)}, \overline{T}^{(2)})$ is the optimal solution, the inequality $\overline{T}^{(1)} \geq \overline{T}^{(2)}$ holds, which proves the theorem.
\end{proof}

\begin{theorem}
	Define
	\begin{align}
		\nonumber & F_{\rm total}^{\rm lim} = \left(\sum_{u=1}^{U} \rho_u D_u \right) \\
		& \hspace{24mm} \min \left\{ \frac{B_{\mathrm{total}}}{\sum_{u=1}^{U} \frac{D_u}{r_u}} , \frac{R^S_{\mathrm{total}}}{\sum_{u=1}^{U} \zeta_u D_u} \right\}, \label{F lim}
	\end{align}
	and
	\begin{equation}
		T^{\rm lim} = \max \left\{ \frac{\sum_{u=1}^{U} \frac{D_u}{r_u}}{B_{\mathrm{total}}} , \frac{\sum_{u=1}^{U} \zeta_u D_u}{R^S_{\mathrm{total}}} \right\}. \label{T lim}
	\end{equation}
	When the following inequality holds
	\begin{equation}
		F_{\mathrm{total}} \geq F_{\rm total}^{\rm lim}, \label{F cond}
	\end{equation}
	the closed-form expressions of the optimal solution to problem (\ref{problem 4}) are given as
	\begin{align}
		\overline{T} & = T^{\rm lim}, \\
		\nonumber \overline{B_{u}} & = \frac{D_u}{r_u} \min \left\{ \frac{B_{\mathrm{total}}}{\sum_{u=1}^{U} \frac{D_u}{r_u}} , \frac{R^S_{\mathrm{total}}}{\sum_{u=1}^{U} \zeta_u D_u} \right\}, \\
		& \hspace{50mm} \forall u = 1,...,U, \\
		\nonumber \overline{R^S_{u}} & = \zeta_u D_u \min \left\{ \frac{B_{\mathrm{total}}}{\sum_{u=1}^{U} \frac{D_u}{r_u}} , \frac{R^S_{\mathrm{total}}}{\sum_{u=1}^{U} \zeta_u D_u} \right\}, \\
		& \hspace{50mm} \forall u = 1,...,U.
	\end{align}
\end{theorem}
\begin{proof}
	See Appendix D.
\end{proof}

Consider the following resource configuration problem: Given that the configured total user-UAV communication bandwidth $B_{\mathrm{total}}$ and total UAV-satellite data rate $R_{\mathrm{total}}^S$ are determined, to ensure that the minimized overall latency does not exceed a threshold $T^{\rm th}$, how much total computing capability $F_{\rm total}$ needs to be configured?

Based on the Theorem 4 and Theorem 5 presented above, we propose the following resource configuration
principles:
\begin{itemize}
	\item If $T^{\rm th} < T^{\rm lim}$, it is impossible to finish the data upload within the threshold latency, no matter how much total computing capability $F_{\rm total}$ is configured.
	\item If $T^{\rm th} \geq T^{\rm lim}$, configuring the total computing capability as $F_{\rm total} = F_{\rm total}^{\rm lim}$ is a sufficient condition to finish the data upload within the latency threshold, where $T^{\rm lim}$ is the corresponding overall latency. Further increasing $F_{\rm total}$ could not reduce the overall latency.
\end{itemize} 

In the second scenario, the configuration scheme $F_{\rm total} = F_{\rm total}^{\rm lim}$ is sufficient but probably not optimal, which could lead to resource redundancy. The optimal resource configuration problem is complicated and requires further research.

\begin{figure}[b]
	\centering
	\includegraphics[width=3.5in]{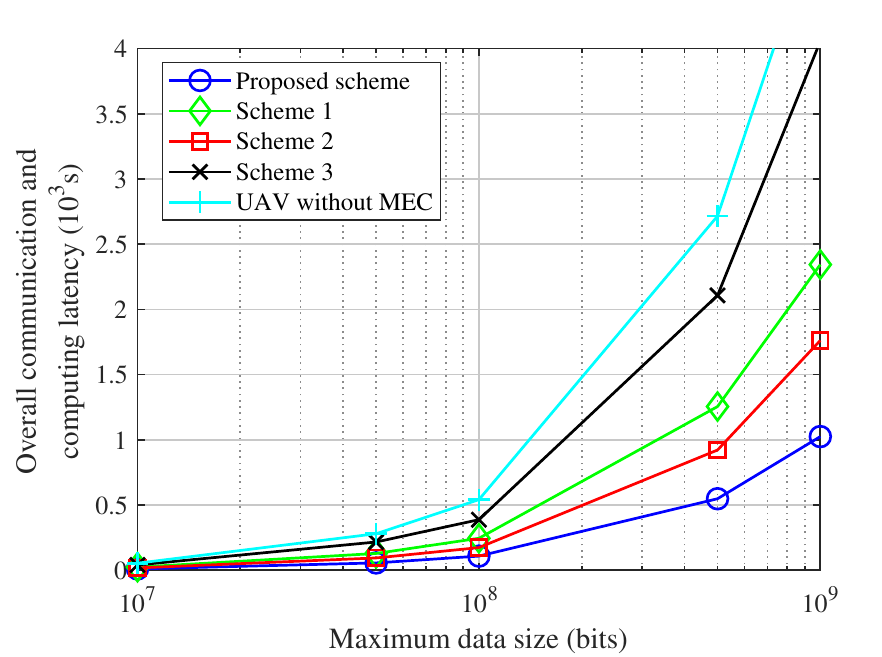}
	\caption{Overall latency comparison among different algorithms with different maximum data size.}
	\label{main comp small}
\end{figure}

\begin{figure*}[t]
	\centering
	\subfloat[]{
		\includegraphics[width=0.5\linewidth]{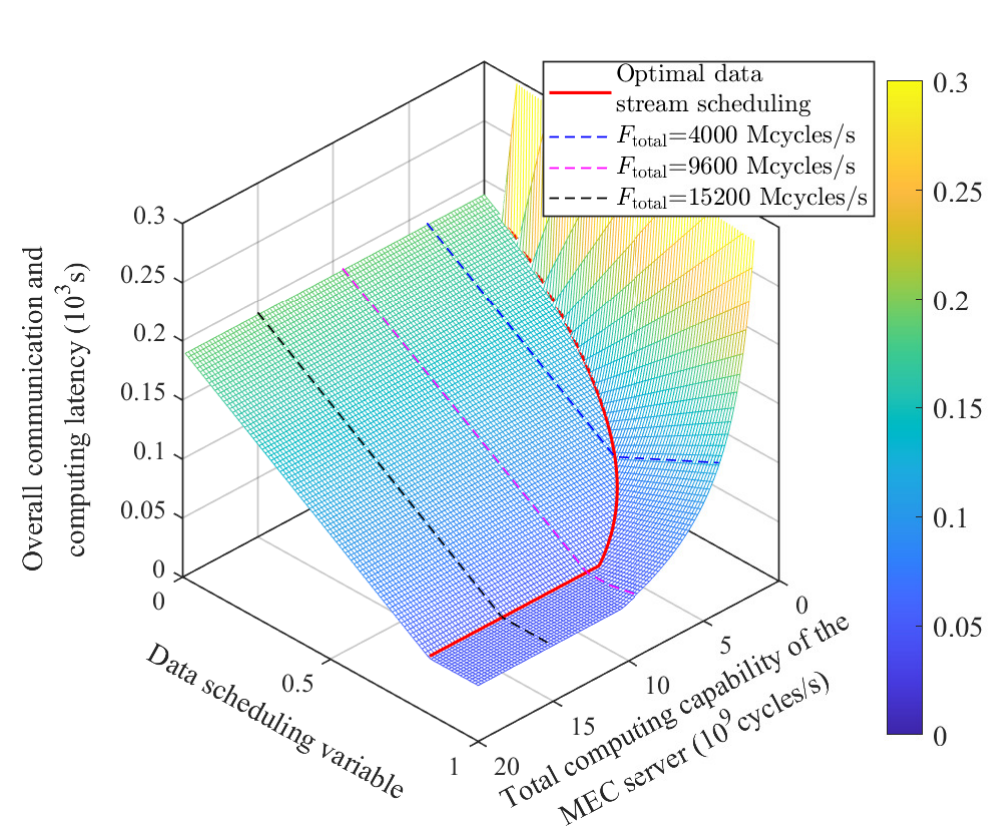}
	}\hfill
	\subfloat[]{
		\includegraphics[width=0.46\linewidth]{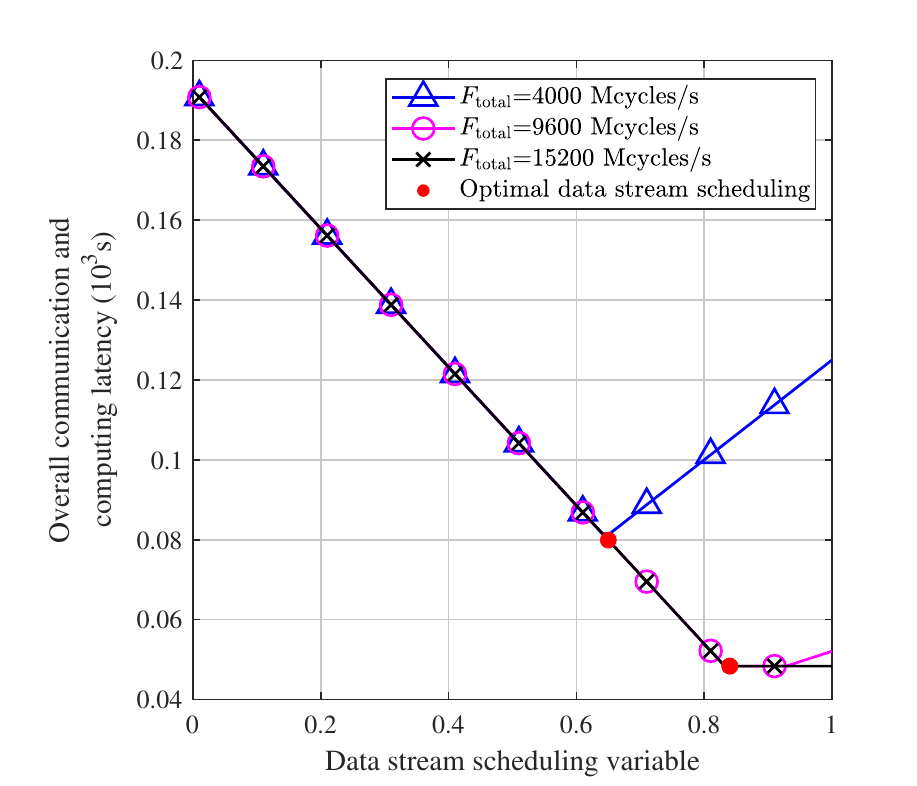}
	}\\
	\vspace{-3mm}
	\subfloat[]{
		\includegraphics[width=0.5\linewidth]{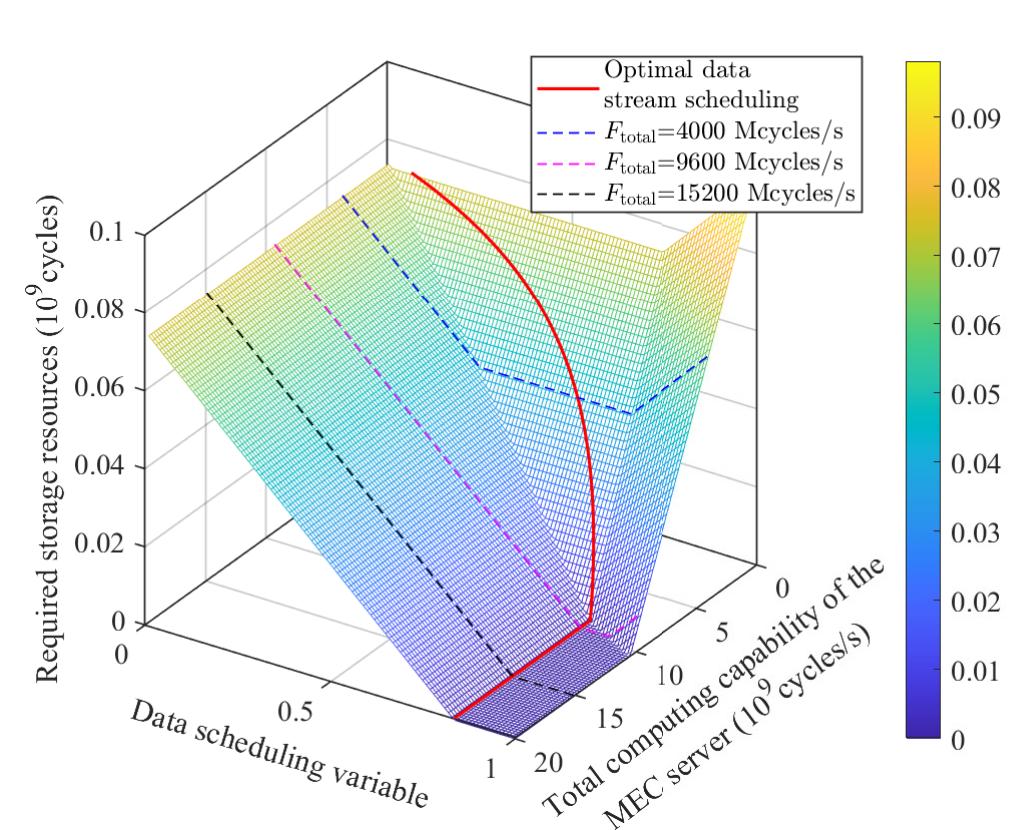}
	}\hfill
	\subfloat[]{
		\includegraphics[width=0.46\linewidth]{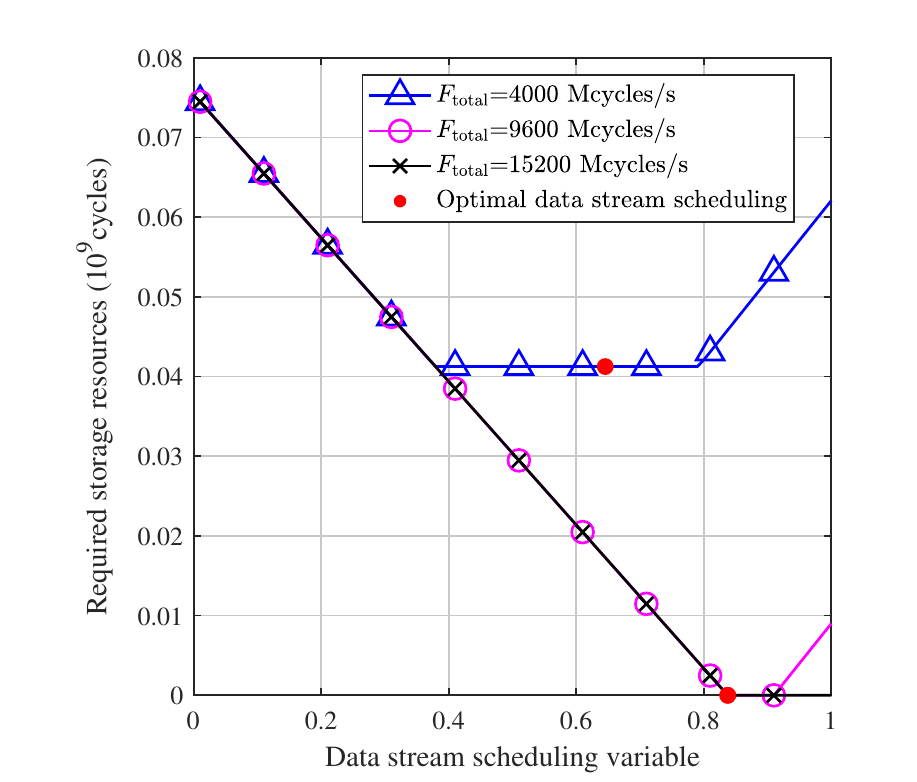}
	}
	\caption{Overall latency and the minimum storage required varying with the data scheduling variable.}
	\label{eta_opt_fig}
\end{figure*}

\section{Simulation Results}
We evaluate the performance of our proposed scheme through simulations in this section. The following simulation parameters are used unless otherwise specified. The number of users is set as $U = 4$. For each user's sensing data, the data size $D_u$ is uniformly distributed in range $[\frac{1}{10}D_{\rm max},D_{\rm max}]$, where the maximum data size is set as $D_{\rm max} = 100$ Mbits. Besides, the computing intensity $\rho_u$ is uniformly distributed in range $[1000,5000]$ cycles/bit \cite{sim 01}, and the computing output-to-input ratio is uniformly distributed in range $[0.01,0.1]$ \cite{sim 02}. We assume that the position of the UAV is $[0,0,1000]^T$ m, and the users are uniformly distributed in a cycle denoted as $\{[x,y,0]^T|\sqrt{x^2+y^2} \leq 1000(m)\}$. For the user-UAV transmission, the channel parameters are set as $(\eta_{\rm LoS},\eta_{\rm NLoS},a,b) = (0.1,21,5.0188,0.3511)$ \cite{sim 03}. The carrier frequency $f$ is set to be $f = 5.8$ GHz, and the speed of light is $c = 3\times10^8$ m/s. The transmit power for each user is set to be $p_u = 1$ W and the noise power is $\sigma^2 = -114$ dBm. The total bandwidth is set as $B_{\rm total} = 0.5$ MHz. The ergodic spectrum efficiency between user $u$ and the UAV $r_u$ is obtained by averaging over 1000 generated small-scale channels. The total UAV-satellite transmission rate is set as $R_{\rm total}^S = 0.5$ Mbits/s. Moreover, the total computing capability of the computing unit is set as $F_{\rm total} = 5 \times 10^9$ cycles/s.

We first compare the proposed scheme with other different algorithms. A simple scheme is considered first where the UAV is not equipped with an MEC server. In addition to that, three other schemes are considered as
\begin{itemize}
	\item Scheme 1: We consider the algorithm proposed in \cite{RW-MEC 10}.
	\item Scheme 2: We consider a simplified version of our proposed algorithm, where we equally allocate the computing and communication resources to all users.
	\item Scheme 3: A simplified version of the algorithm proposed in \cite{RW-MEC 10} is used, where the computing and communication resources are allocated equally to all users.
\end{itemize}
In this simulation, we select a series of maximum data size $D_{\rm max}$ from $0.1$ Mbits to $10$ Mbits. The latency results are averaged over 50 randomly-selected network topologies. As shown in Fig. \ref{main comp small}, the proposed scheme could obtain better system performance than the scheme proposed in \cite{RW-MEC 10}. The reason is that in \cite{RW-MEC 10} the overall latency is the sum of the user-UAV transmission latency, the computing latency and the UAV-satellite transmission latency, while in our proposed scheme, the transmission and computing processes could be carried out simultaneously to decrease the overall latency. Also, by comparing the latency obtained by the proposed scheme and Scheme 2, one concludes that the resource orchestration optimization provides a significant performance gain in terms of upload latency.

\begin{figure*}[t]
	\centering
	\subfloat[]{
		\includegraphics[width=0.5\linewidth]{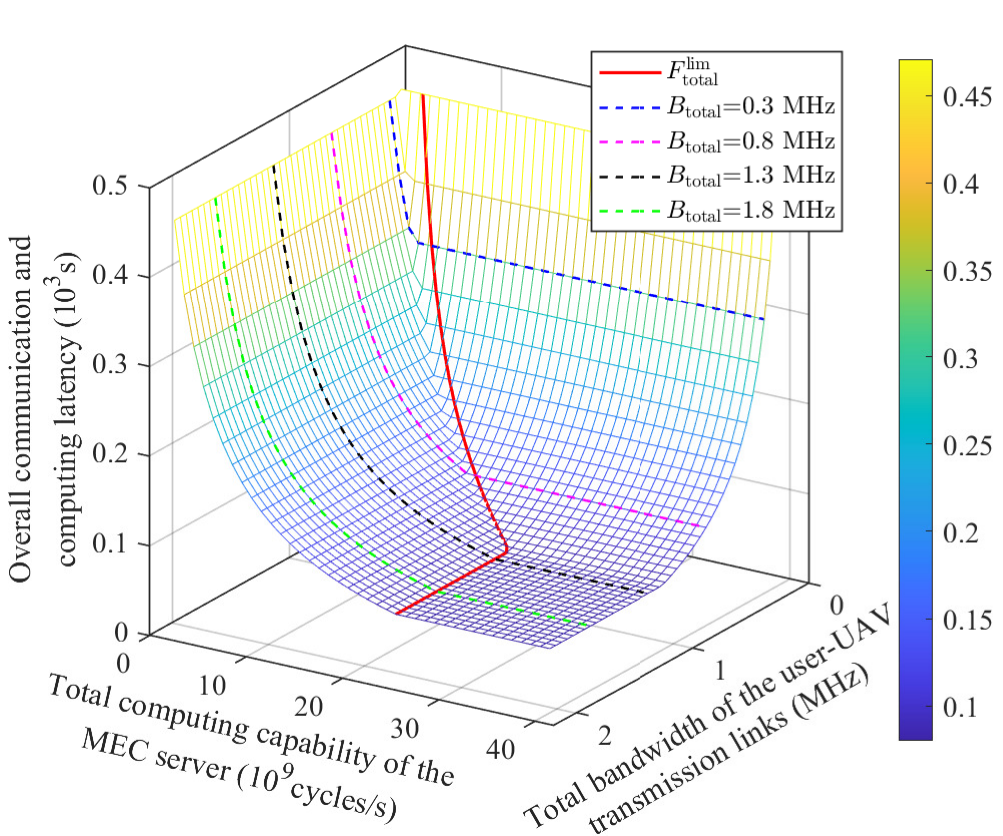}
	}\hfill
	\subfloat[]{
		\includegraphics[width=0.42\linewidth]{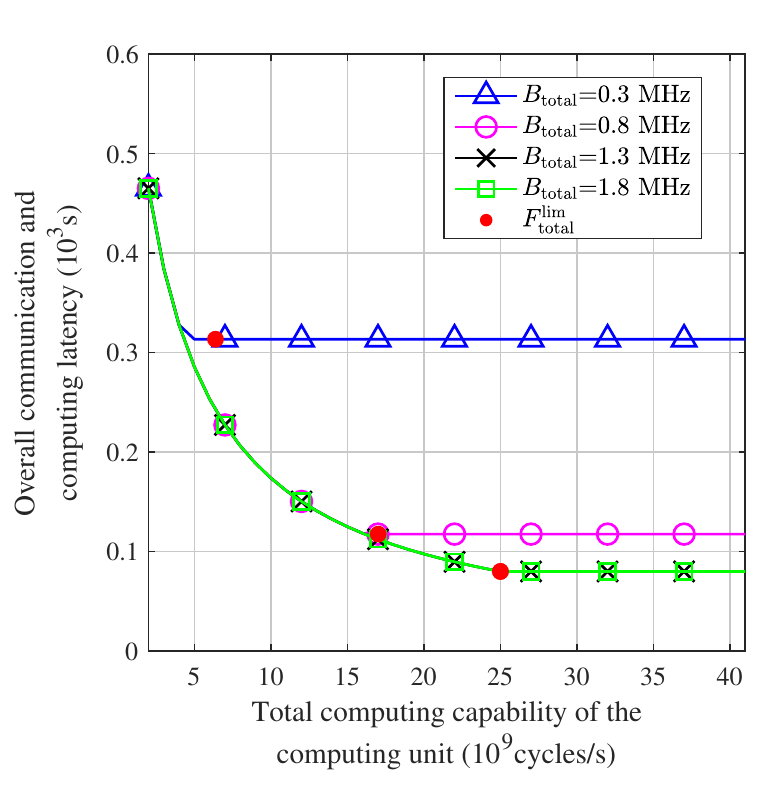}
	}
	\caption{Sufficient condition of the configured total computing capability and corresponding overall upload latency.}
	\label{main F total}
\end{figure*}

In Fig. \ref{eta_opt_fig}, we verify the performance of our proposed data scheduling scheme. Specifically, we consider the scenario of a single user, where both user-UAV bandwidth and UAV-satellite data rate are fixed ($B_{\rm total} = 0.5$ MHz, $R_{\rm total}^S = 0.5$ Mbits/s). With the total computing capability $F_{\rm total}$ varying between $[0.2,20]\times10^9$ cycles/s, the relationship between the overall latency and the data scheduling variable is simulated, as shown in Fig. \ref{eta_opt_fig}(a). The data scheduling variables obtained through our proposed method are also calculated and presented in Fig. \ref{eta_opt_fig}(a). It can be observed that the data scheduling variable obtained by the proposed scheme always achieves the minimum latency at any given $F_{\rm total}$. To show the results more clearly, we select three specific $F_{\rm total}$ values to sketch the overall latency with varying scheduling variables, as shown in Fig. \ref{eta_opt_fig}(b). We can see that when $F_{\rm total}$ is relatively small, our proposed scheme obtains the exact optimal scheduling variable value to achieve the minimum latency. As $F_{\rm total}$ increases, variables in a certain interval are all optimal, and our proposed scheme always select a certain value in the interval. Similarly, Fig. \ref{eta_opt_fig}(c) and Fig. \ref{eta_opt_fig}(d) present the relationship between the required storage and the scheduling variable. We could also observe that our proposed scheme always obtain the scheduling variable that achieves the minimum required storage. We also notice that this minimized required storage appears to be 0, which means that if the data is properly scheduled, the system does not require extra storage resources for stranded data.

In Fig. \ref{main F total}, we verify the proposed principles for configuring the total computing capability. The total UAV-satellite data rate is set as $R_{\rm total}^S = 0.5$ Mbits/s. As shown in Fig. \ref{main F total}(a), we depict the minimized overall latency obtained with varying user-UAV total bandwidth $B_{\rm total}$ and total computing capability $F_{\rm total}$, as well as the latency obtained with $F_{\rm total}^{\rm lim}$. Four example values of $B_{\rm total}$ are selected to show the specific relationship in Fig. \ref{main F total}(b). We could see that for any given $B_{\rm total}$, when $F_{\rm total}$ is larger than $F_{\rm total}^{\rm lim}$, the overall latency does not decline as $F_{\rm total}$ increases, which verifies our proposition. Besides, we notice from the $B_{\rm total} = 0.3$ MHz scenario that the latency may stop declining before $F_{\rm total}$ reaches $F_{\rm total}^{\rm lim}$, which implies that further optimization of the configured total computing capability is possible. This requires further research into this problem.


Finally, in Fig. \ref{main B_S total} we present the impact of the UAV-satellite total data rate $R^S_{\mathrm{total}}$ on the overall latency with varying user-UAV total transmission bandwidth $B_{\mathrm{total}}$. The latency results are averaged over 50 randomly-selected network topologies. We can observe that as $R^S_{\mathrm{total}}$ increases the overall latency first reduces to a minimum and then remains unchanged. This is because when $B^S_{\mathrm{total}}$ is large enough the other two parameters become the system performance bottleneck. Besides, as $R_{\mathrm{total}}$ increases the overall latency decreases and eventually reaches a minimum value. Similarly, the reason is $B_{\mathrm{total}}$ becomes large enough and is no longer the bottleneck of the system performance.

\begin{figure}[t]
	\centering
	\includegraphics[width=3.5in]{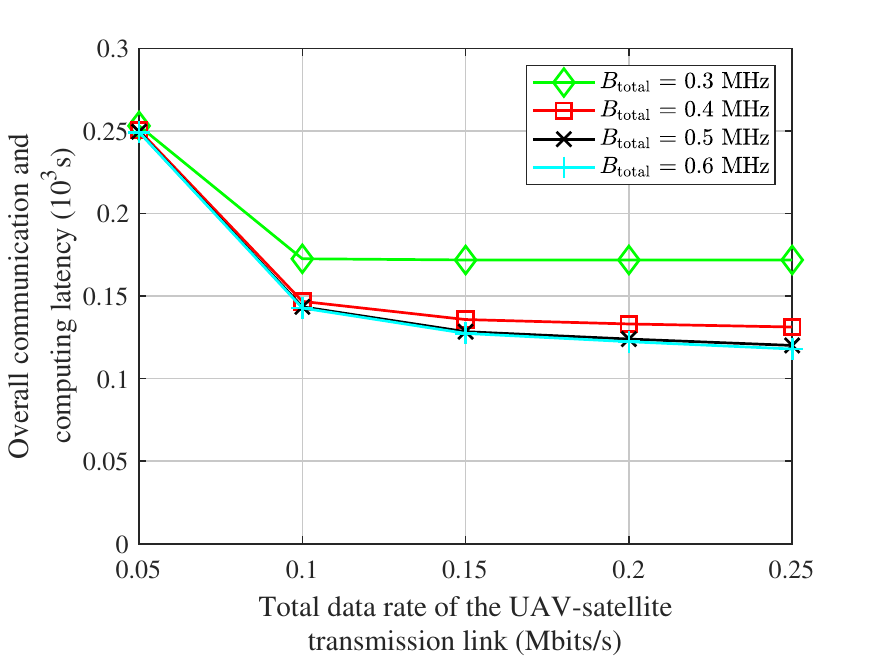}
	\caption{Relationship between the overall latency and the UAV-satellite transmission bandwidth.}
	\label{main B_S total}
\end{figure}

\section{Conclusions}
In this paper, we have modeled the EIH-empowered NTN to support low-latency sensing data upload for quick response to disasters. The EIH incorporates communications, computing and storage capabilities, which is envisioned to synergize heterogeneous coupling parts, and enable systematic design. We have investigated two major design problems of the EIH-empowered NTN to improve the latency performance and resource efficiency. First, the joint data scheduling and resource orchestration problem has been formulated for latency minimization. The problem has been transformed equivalently into a convex one. Thus, an optimal joint data scheduling and resource orchestration scheme has been proposed. On that basis, we have further derived the principles for resource configuration of the EIH, under its payload limitations in terms of size, weight and energy supply. Simulation results have corroborated our theoretical achievements, and have also demonstrated the superiority of our proposed scheme in reducing the overall upload latency, thus enabling quick emergency rescue.

\appendices
\section{Proof of \textit{Proposition 1}}
In order to obtain the expressions of the functions $T_u(B_u,R_u^S,F_u,\eta_u)$ and $V_u(B_u,R_u^S,F_u,\eta_u)$, we take the data flow perspective as shown in Fig. 2. We present five different situations based on the relationship among the data flow rates, and discuss the $T_u$ and $V_u$ expressions of each situation.

\subsection{The amount of to-be-computed data and to-be-uploaded data in the storage both grow during the user-UAV transmission, and the former runs out first afterwards.}
We first focus on the to-be-computed data in the storage unit as shown in Fig. 2. For this part of data, the input data flow is a ratio of the user-UAV transmitted data flow, with the data rate of $\eta_u B_u r_u$. Once the user-UAV transmission is completed, the input data rate becomes zero. The output data flow is the data flow to the computing unit for computation, which has data rate $\frac{F_u}{\rho_u}$. In the first situation, we assume
\begin{equation}
	\eta_u B_u r_u \geq \frac{F_u}{\rho_u}, \label{appd1.a.1}
\end{equation}
which suggests that the amount of to-be-computed data in the storage grows during the user-UAV transmission.

Then we focus on the to-be-uploaded data in the storage unit. The input data flow for this part of data is a ratio of the user-UAV transmitted data flow plus the computed outcome data flow from the computing unit, and the combined data rate is $\frac{\zeta_u F_u}{\rho_u}+(1-\eta_u) B_u r_u$. The input data rate also turns zero after the user-UAV transmission. The output data flow is the UAV-satellite transmitted data flow with the data rate of $R_u^S$. We assume in the first situation that
\begin{equation}
	\frac{\zeta_u F_u}{\rho_u}+(1-\eta_u) B_u r_u \geq R_u^S. \label{appd1.a.2}
\end{equation}
Similarly, this suggests that during the user-UAV transmission the amount of to-be-uploaded data in the storage unit grows.

The duration of user-UAV transmission could be calculated as $\frac{D_u}{B_u r_u}$. By the time the user-UAV transmission is completed, the to-be-computed data in the storage unit have the size of $\frac{D_u}{B_u r_u}(\eta_u B_u r_u - \frac{F_u}{\rho_u})$, and the size of the to-be-uploaded data is $\frac{D_u}{B_u r_u}[\frac{\zeta_u F_u}{\rho_u}+(1-\eta_u) B_u r_u - R_u^S]$. We assume that the former part of data run out first after the user-UAV transmission is completed, which can be expressed as
\begin{align}
	\nonumber & \frac{\frac{D_u}{B_u r_u}(\eta_u B_u r_u - \frac{F_u}{\rho_u})}{\frac{F_u}{\rho_u}} \\
	& \hspace{10mm} \leq \frac{\frac{D_u}{B_u r_u}[\frac{\zeta_u F_u}{\rho_u}+(1-\eta_u) B_u r_u - R_u^S]}{R_u^S - \frac{\zeta_u F_u}{\rho_u}},
\end{align}
which can be further simplified as
\begin{equation}
	\frac{F_u}{\rho_u} \geq \frac{\eta_u R_u^S}{\zeta_u \eta_u + 1 - \eta_u}. \label{appd1.a.3}
\end{equation}

Under these assumptions, the overall communication and computing latency is composed of three parts. The first part is the user-UAV transmission latency where data are accumulated in the storage. The second part is the latency of the to-be-computed data running out, and the third part is the latency of the remaining to-be-uploaded data running out. Therefore, the overall latency $T_u(B_u,R_u^S,F_u,\eta_u)$ can be expressed as
\begin{align}
	\nonumber T_u & = \frac{D_u}{B_u r_u}+\frac{\frac{D_u}{B_u r_u}(\eta_u B_u r_u - \frac{F_u}{\rho_u})}{\frac{F_u}{\rho_u}} \\
	\nonumber & \hspace{5mm} + \frac{1}{R_u^S}\left\{ \frac{D_u}{B_u r_u}\left[\frac{\zeta_u F_u}{\rho_u}+(1-\eta_u) B_u r_u - R_u^S\right] \right. \\
	\nonumber & \left. \hspace{5mm} - (R_u^S - \frac{\zeta_u F_u}{\rho_u})\frac{\frac{D_u}{B_u r_u}(\eta_u B_u r_u - \frac{F_u}{\rho_u})}{\frac{F_u}{\rho_u}} \right\} \\
	& = \frac{D_u}{R_u^S}(\zeta_u \eta_u + 1 - \eta_u). \label{appd1.a.4}
\end{align}
This expression can be understood in a simpler way, namely the UAV-satellite transmission data rate is the bottleneck of the system, and thus the overall latency is the total data to be transmitted from the UAV to the satellite $D_u(\zeta_u \eta_u + 1 - \eta_u)$ divided by the UAV-satellite transmission rate $R_u^S$.

We further consider the minimum storage $V_u(B_u,R_u^S,F_u,\eta_u)$ required. It can be observed that the storage unit accumulates data during the user-UAV transmission, and afterwards the amount of data decreases because of the computation and transmission to the satellite. Therefore, the minimum storage $V_u(B_u,R_u^S,F_u,\eta_u)$ required corresponds to the moment at which the user-UAV transmission is just completed, and the expression is
\begin{equation}
	V_u = \frac{D_u}{B_u r_u}\left[B_u r_u - R_u^S - (1-\zeta_u)\frac{F_u}{\rho_u}\right]. \label{appd1.a.5}
\end{equation}

In this first situation, we prove that under assumptions (\ref{appd1.a.1}), (\ref{appd1.a.2}) and (\ref{appd1.a.3}), the expressions for $T_u(B_u,R_u^S,F_u,\eta_u)$ and $V_u(B_u,R_u^S,F_u,\eta_u)$ are shown in (\ref{appd1.a.4}) and (\ref{appd1.a.5}), which corresponds to (\ref{expr_T}a) and (\ref{expr_V}a).

\subsection{The amount of to-be-computed data and to-be-uploaded data in the storage both grow during the user-UAV transmission, and the latter runs out first afterwards.}
In the second situation, we assume that (\ref{appd1.a.1}) and (\ref{appd1.a.2}) still hold. This means that during the user-UAV transmission, both the to-be-computed data and the to-be-uploaded data are accumulating in the storage unit.

Different from the first situation, we assume that the to-be-uploaded data to the satellite in the storage run out first, which means that the following inequality holds
\begin{equation}
	\frac{F_u}{\rho_u} < \frac{\eta_u R_u^S}{\zeta_u \eta_u + 1 - \eta_u}. \label{appd1.b.1}
\end{equation}

Under the assumptions (\ref{appd1.a.1}), (\ref{appd1.a.2}) and (\ref{appd1.b.1}), the overall communication and computing latency include the latency of user-UAV transmission and the latency of the to-be-computed data running out in the storage. In this case, the expression of $T_u(B_u,R_u^S,F_u,\eta_u)$ is given by
\begin{align}
	\nonumber T_u & = \frac{D_u}{B_u r_u}+\frac{\frac{D_u}{B_u r_u}(\eta_u B_u r_u - \frac{F_u}{\rho_u})}{\frac{F_u}{\rho_u}} \\
	& = \frac{\eta_u D_u \rho_u}{F_u}. \label{appd1.b.2}
\end{align}
This expression shows that the computation capability is the bottleneck of this system, and the overall latency is the total data required to be computed $\eta_u D_u$ divided by the computing unit throughput $F_u/\rho_u$.

Similar to the first situation, the data in the storage unit achieve maximum amount at the moment when the user-UAV transmission is completed, and therefore the expression of the minimum storage required $V_u(B_u,R_u^S,F_u,\eta_u)$ in this situation is the same as (\ref{appd1.a.5}).

We prove that in the second situation, under assumptions (\ref{appd1.a.1}), (\ref{appd1.a.2}) and (\ref{appd1.b.1}), the expressions for $T_u(B_u,R_u^S,F_u,\eta_u)$ and $V_u(B_u,R_u^S,F_u,\eta_u)$ are shown as in (\ref{appd1.b.2}) and (\ref{appd1.a.5}), which corresponds to (\ref{expr_T}b) and (\ref{expr_V}b).

\subsection{Only the amount of to-be-computed data in the storage grow during the user-UAV transmission.}
In the third situation, we assume that inequality (\ref{appd1.a.1}) still holds, which means that the amount of to-be-computed data are increasing during the user-UAV transmission.

Nevertheless, we assume that instead of (\ref{appd1.a.2}), the following inequality holds in this situation
\begin{equation}
	\frac{\zeta_u F_u}{\rho_u}+(1-\eta_u) B_u r_u < R_u^S, \label{appd1.c.1}
\end{equation}
which means that the output data rate is larger than the input data rate for the part of to-be-uploaded data in the storage unit. In this case, the actual UAV-satellite data rate (output data rate) is the same as the partial user-UAV transmitted data rate and the computed outcome data rate combined (input data rate) during the user-UAV transmission. Afterwards the actual UAV-satellite data rate becomes the same as the computed outcome data rate. The amount of to-be-uploaded data to the satellite in the storage unit is always equal to zero in this situation.

Similar to the second situation, the computation rate is the bottleneck of the system since the whole process is finished when the to-be-computed data in the storage run out. The expression of $T_u(B_u,R_u^S,F_u,\eta_u)$ is the same as (\ref{appd1.b.2}).

Maximum data amount in the storage is achieved when the user-UAV transmission is just finished, and the required storage size $T_u(B_u,R_u^S,F_u,\eta_u)$ can be given by
\begin{equation}
	V_u = \frac{D_u}{B_u r_u}\left(\eta_u B_u r_u - \frac{F_u}{\rho_u}\right). \label{appd1.c.2}
\end{equation}
The minimum required storage in this situation is different from the previous two situations because the actual UAV-satellite data rate changes to match the input data rate of the to-be-uploaded data in the storage.

In this subsection, we prove that in the third situation, under assumptions (\ref{appd1.a.1}) and (\ref{appd1.c.1}), the expressions for $T_u(B_u,R_u^S,F_u,\eta_u)$ and $V_u(B_u,R_u^S,F_u,\eta_u)$ are shown as in (\ref{appd1.b.2}) and (\ref{appd1.c.2}), which corresponds to (\ref{expr_T}c) and (\ref{expr_V}c).

\subsection{Only the amount of to-be-uploaded data in the storage grow during the user-UAV transmission.}
We assume in the fourth situation that the following inequality instead of (\ref{appd1.a.1}) holds
\begin{equation}
	\eta_u B_u r_u < \frac{F_u}{\rho_u}, \label{appd1.d.1}
\end{equation}
which means that the output data rate is larger than the input data rate for the to-be-computed data in the storage unit. In this case, the actual data rate of the computing unit obtaining data from the storage (output data rate) is the same as the partial user-UAV transmitted data rate (input data rate) during the user-UAV transmission. The amount of to-be-computed data in the storage always equals to zero.

We also assume that the following inequality holds
\begin{equation}
	(\zeta_u\eta_u+1-\eta_u) B_u r_u \geq R_u^S, \label{appd1.d.2}
\end{equation}
which means the input data rate for the to-be-uploaded data in the storage unit is larger than the output data rate and the data amount is growing during the user-UAV transmission. It should be noted that (\ref{appd1.d.2}) is different from (\ref{appd1.a.2}) because the computing unit throughput changes in this situation.

Because the whole process ends when the to-be-uploaded data in the storage unit run out, the system is limited by the UAV-satellite transmission rate, and the overall communication and computation latency $T_u(B_u,R_u^S,F_u,\eta_u)$ is the same as (\ref{appd1.a.4}).

The data amount in the storage also reaches the maximum value when the user-UAV transmission is just finished, and the minimum storage required $V_u(B_u,R_u^S,F_u,\eta_u)$ can be given by
\begin{equation}
	V_u = \frac{D_u}{B_u r_u}[(\zeta_u\eta_u+1-\eta_u) B_u r_u - R_u^S]. \label{appd1.d.3}
\end{equation}

In the fourth situation, we prove that under assumptions (\ref{appd1.d.1}) and (\ref{appd1.d.2}), the expressions for $T_u(B_u,R_u^S,F_u,\eta_u)$ and $V_u(B_u,R_u^S,F_u,\eta_u)$ are shown as in (\ref{appd1.a.4}) and (\ref{appd1.d.3}), which corresponds to (\ref{expr_T}d) and (\ref{expr_V}d).

\subsection{Neither the amount of to-be-computed data nor the amount of to-be-uploaded data in the storage grow during the user-UAV transmission.}
In the fifth situation, we assume that (\ref{appd1.d.1}) still holds which suggests that the amount of to-be-computed data in the storage unit is always zero.

We also assume that the following inequality instead of (\ref{appd1.d.2}) holds
\begin{equation}
	(\zeta_u\eta_u+1-\eta_u) B_u r_u < R_u^S, \label{appd1.e.1}
\end{equation}
which means the input data rate for the to-be-uploaded data in the storage unit is smaller than the output data rate. In this case the amount of the to-be-uploaded data in the storage also always equals zero.

In this situation, no amount of data are accumulating in the storage. Thus, the whole process is finished when the user-UAV transmission is completed, and the overall communication and computation latency $T_u(B_u,R_u^S,F_u,\eta_u)$ is
\begin{equation}
	T_u = \frac{D_u}{B_u r_u}, \label{appd1.e.2}
\end{equation}
and the minimum required storage $V_u(B_u,R_u^S,F_u,\eta_u)$ is
\begin{equation}
	V_u = 0. \label{appd1.e.3}
\end{equation}

We prove that under assumptions (\ref{appd1.d.1}) and (\ref{appd1.e.1}), the expressions for $T_u(B_u,R_u^S,F_u,\eta_u)$ and $V_u(B_u,R_u^S,F_u,\eta_u)$ are shown as in (\ref{appd1.e.2}) and (\ref{appd1.e.3}), which corresponds to (\ref{expr_T}e) and (\ref{expr_V}e).

\section{Proof of \textit{Theorem 1}}
Assuming that the values of variables $B_u$, $R_u^S$ and $F_u$ are given, we present three different situations based on the relationship between $B_u$ and $R_u^S$. The optimal value of variable $\eta_u^{\mathrm{opt}}$ and the corresponding minimized function values $T_u^{\eta-opt}(B_u,R_u^S,F_u) = T_u(B_u,R_u^S,F_u,\eta_u^{\mathrm{opt}})$ and $V_u^{\eta-\mathrm{opt}}(B_u,R_u^S,F_u) = V_u(B_u,R_u^S,F_u,\eta_u^{\mathrm{opt}})$ are discussed in each situation.

\subsection{$B_u r_u < R_u^S$}
Under this assumption, we further discuss the relationship between $T_u$ and $\eta_u$ and the relationship between $V_u$ and $\eta_u$ given different $F_u$ values.

We first assume
\begin{equation}
	\frac{F_u}{\rho_u} < B_u r_u.
\end{equation}
In this case, when the following inequality holds
\begin{equation}
	0 \leq \eta_u < \frac{F_u}{\rho_u B_u r_u},
\end{equation}
the expressions of $T_u$ and $V_u$ refer to (\ref{expr_T}e) and (\ref{expr_V}e), where $T_u$ and $V_u$ remain unchanged with $\eta_u$ increasing. Besides, when the following inequality holds
\begin{equation}
	\frac{F_u}{\rho_u B_u r_u} \leq \eta_u \leq 1,
\end{equation}
the expressions of $T_u$ and $V_u$ refer to (\ref{expr_T}c) and (\ref{expr_V}c), where both $T_u$ and $V_u$ monotonically increase with $\eta_u$ increasing. Therefore, the optimal $\eta_u$ could be arbitrarily taken in the range of $[ 0, \frac{F_u}{\rho_u B_u r_u} )$.

Then we assume
\begin{equation}
	\frac{F_u}{\rho_u} \geq B_u r_u.
\end{equation}
In this case, for
\begin{equation}
	0 \leq \eta_u \leq 1,
\end{equation}
the expressions of $T_u$ and $V_u$ refer to (\ref{expr_T}e) and (\ref{expr_V}e), where $T_u$ and $V_u$ remain unchanged with $\eta_u$ increasing. Therefore, the optimal $\eta_u$ could be arbitrarily taken in the range of $[ 0, 1 ]$.

Therefore, on the condition of
\begin{equation}
	B_u r_u < R_u^S,
\end{equation}
we set the optimal value of $\eta_u$ to be
\begin{equation}
	\eta_u^{\mathrm{opt}} = 0,
\end{equation}
for the simplicity of expression. The corresponding function values are
\begin{align}
	T_u^{\eta-\mathrm{opt}}(B_u,R_u^S,F_u) & = \frac{D_u}{B_u r_u}, \\
	V_u^{\eta-\mathrm{opt}}(B_u,R_u^S,F_u) & = 0.
\end{align}
This verifies the first row of Table \ref{eta-opt}.

\subsection{$R_u^S \leq B_u r_u < \frac{R_u^S}{\zeta_u}$}
we also discuss the relationship between $T_u$ and $\eta_u$ and the relationship between $V_u$ and $\eta_u$ given different $F_u$ values as follows.

We first assume
\begin{equation}
	\frac{F_u}{\rho_u} < \frac{B_u r_u - R_u^S}{1-\zeta_u}.
\end{equation}
In this case, when the following inequality holds
\begin{equation}
	0 \leq \eta_u < \frac{F_u}{\rho_u B_u r_u},
\end{equation}
the expressions of $T_u$ and $V_u$ refer to (\ref{expr_T}d) and (\ref{expr_V}d), where both $T_u$ and $V_u$ monotonically decrease with $\eta_u$ increasing. When the following inequality holds
\begin{equation}
	\frac{F_u}{\rho_u B_u r_u} \leq \eta_u < \frac{F_u}{\rho_u R_u^S+(1-\zeta_u)F_u},
\end{equation}
the expressions of $T_u$ and $V_u$ refer to (\ref{expr_T}a) and (\ref{expr_V}a), and with $\eta_u$ increasing, $T_u$ monotonically decrease while $V_u$ remain unchanged. When the following inequality holds
\begin{align}
	\nonumber & \frac{F_u}{\rho_u R_u^S+(1-\zeta_u)F_u} \leq \eta_u \\
	& \hspace{30mm} < \frac{\rho_u(B_u r_u - R_u^S)+\zeta_u F_u}{\rho_u B_u r_u},
\end{align}
the expressions of $T_u$ and $V_u$ refer to (\ref{expr_T}b) and (\ref{expr_V}b), and with $\eta_u$ increasing, $T_u$ monotonically increase while $V_u$ remain unchanged. Finally, when the following inequality holds
\begin{equation}
	\frac{\rho_u(B_u r_u - R_u^S)+\zeta_u F_u}{\rho_u B_u r_u} \leq \eta_u \leq 1,
\end{equation}
the expressions of $T_u$ and $V_u$ refer to (\ref{expr_T}c) and (\ref{expr_V}c), where both $T_u$ and $V_u$ monotonically increase with $\eta_u$ increasing. Therefore, on the condition of
\begin{equation}
	R_u^S \leq B_u r_u < \frac{R_u^S}{\zeta_u},\ \frac{F_u}{\rho_u} < \frac{B_u r_u - R_u^S}{1-\zeta_u}
\end{equation}
the optimal value of $\eta_u$ to minimize both functions is determined to be
\begin{equation}
	\eta_u^{\mathrm{opt}} = \frac{F_u}{\rho_u R_u^S+(1-\zeta_u)F_u},
\end{equation}
and the corresponding function values are
\begin{align}
	T_u^{\eta-\mathrm{opt}}(B_u,R_u^S,F_u) & = \frac{\rho_u D_u}{\rho_u R_u^S+(1-\zeta_u)F_u}, \\
	\nonumber V_u^{\eta-\mathrm{opt}}(B_u,R_u^S,F_u) & = \frac{D_u}{B_u r_u}\left[B_u r_u - R_u^S \right.\\
	& \left. \hspace{20mm} - (1-\zeta_u) \frac{F_u}{\rho_u}\right],
\end{align}
This verifies the second row of Table \ref{eta-opt}.

Then we assume
\begin{equation}
	\frac{B_u r_u - R_u^S}{1-\zeta_u} \leq \frac{F_u}{\rho_u} < B_u r_u.
\end{equation}
In this case, when the following inequality holds
\begin{equation}
	0 \leq \eta_u < \frac{B_u r_u - R_u^S}{(1-\zeta_u)B_u r_u},
\end{equation}
the expressions of $T_u$ and $V_u$ refer to (\ref{expr_T}d) and (\ref{expr_V}d), where both $T_u$ and $V_u$ monotonically decrease with $\eta_u$ increasing. When the following inequality holds
\begin{equation}
	\frac{B_u r_u - R_u^S}{(1-\zeta_u)B_u r_u} \leq \eta_u < \frac{F_u}{\rho_u B_u r_u},
\end{equation}
the expressions of $T_u$ and $V_u$ refer to (\ref{expr_T}e) and (\ref{expr_V}e), where $T_u$ and $V_u$ remain unchanged with $\eta_u$ increasing. When the following inequality holds
\begin{equation}
	\frac{F_u}{\rho_u B_u r_u} \leq \eta_u \leq 1,
\end{equation}
the expressions of $T_u$ and $V_u$ refer to (\ref{expr_T}c) and (\ref{expr_V}c), where both $T_u$ and $V_u$ monotonically increase with $\eta_u$ increasing. Therefore, the optimal $\eta_u$ could be arbitrarily taken in the range of $[ \frac{B_u r_u - R_u^S}{(1-\zeta_u)B_u r_u}, \frac{F_u}{\rho_u B_u r_u} )$.

Finally we assume
\begin{equation}
	\frac{F_u}{\rho_u} \geq B_u r_u.
\end{equation}
In this case, when the following inequality holds
\begin{equation}
	0 \leq \eta_u < \frac{B_u r_u - R_u^S}{(1-\zeta_u)B_u r_u},
\end{equation}
the expressions of $T_u$ and $V_u$ refer to (\ref{expr_T}d) and (\ref{expr_V}d), where both $T_u$ and $V_u$ monotonically decrease with $\eta_u$ increasing. When the following inequality
\begin{equation}
	\frac{B_u r_u - R_u^S}{(1-\zeta_u)B_u r_u} \leq \eta_u \leq 1,
\end{equation}
the expressions of $T_u$ and $V_u$ refer to (\ref{expr_T}e) and (\ref{expr_V}e), where $T_u$ and $V_u$ remain unchanged with $\eta_u$ increasing. Therefore, the optimal $\eta_u$ could be arbitrarily taken in the range of $[ \frac{B_u r_u - R_u^S}{(1-\zeta_u)B_u r_u}, 1]$.

Therefore, on the condition of
\begin{equation}
	R_u^S \leq B_u r_u < \frac{R_u^S}{\zeta_u}, \frac{F_u}{\rho_u} \geq \frac{B_u r_u - R_u^S}{1-\zeta_u},
\end{equation}
we set the optimal value of $\eta_u$ to be
\begin{equation}
	\eta_u^{\mathrm{opt}} = \frac{B_u r_u - R_u^S}{(1-\zeta_u)B_u r_u},
\end{equation}
for the simplicity of expression. The corresponding function values are
\begin{align}
	T_u^{\eta-\mathrm{opt}}(B_u,R_u^S,F_u) & = \frac{D_u}{B_u r_u}, \\ V_u^{\eta-\mathrm{opt}}(B_u,R_u^S,F_u) & = 0.
\end{align}
This verifies the third row of Table \ref{eta-opt}.

\subsection{$B_u r_u \geq \frac{R_u^S}{\zeta_u}$}
We first assume
\begin{equation}
	\frac{F_u}{\rho_u} < \frac{R_u^S}{\zeta_u}.
\end{equation}
In this case, when the following inequality holds
\begin{equation}
	0 \leq \eta_u < \frac{F_u}{\rho_u B_u r_u},
\end{equation}
the expressions of $T_u$ and $V_u$ refer to (\ref{expr_T}d) and (\ref{expr_V}d), where both $T_u$ and $V_u$ monotonically decrease with $\eta_u$ increasing. When the following inequality holds
\begin{equation}
	\frac{F_u}{\rho_u B_u r_u} \leq \eta_u < \frac{F_u}{\rho_u R_u^S+(1-\zeta_u)F_u},
\end{equation}
the expressions of $T_u$ and $V_u$ refer to (\ref{expr_T}a) and (\ref{expr_V}a), and with $\eta_u$ increasing, $T_u$ monotonically decrease while $V_u$ remain unchanged. When the following inequality holds
\begin{align}
	\nonumber & \frac{F_u}{\rho_u R_u^S+(1-\zeta_u)F_u} \leq \eta_u \\
	& \hspace{30mm} < \frac{\rho_u(B_u r_u - R_u^S)+\zeta_u F_u}{\rho_u B_u r_u},
\end{align}
the expressions of $T_u$ and $V_u$ refer to (\ref{expr_T}b) and (\ref{expr_V}b). With $\eta_u$ increasing, $T_u$ monotonically increase while $V_u$ remain unchanged. Finally, when the following inequality holds
\begin{equation}
	\frac{\rho_u(B_u r_u - R_u^S)+\zeta_u F_u}{\rho_u B_u r_u} \leq \eta_u \leq 1,
\end{equation}
the expressions of $T_u$ and $V_u$ refer to (\ref{expr_T}c) and (\ref{expr_V}c), where both $T_u$ and $V_u$ monotonically increase with $\eta_u$ increasing.
Therefore, on the condition of
\begin{equation}
	B_u r_u \geq \frac{R_u^S}{\zeta_u},\ \frac{F_u}{\rho_u} < \frac{R_u^S}{\zeta_u},
\end{equation}
the optimal value of $\eta_u$ to minimize both functions is determined to be
\begin{equation}
	\eta_u^{\mathrm{opt}} = \frac{F_u}{\rho_u R_u^S+(1-\zeta_u)F_u},
\end{equation}
and the corresponding function values are 
\begin{align}
	T_u^{\eta-\mathrm{opt}}(B_u,R_u^S,F_u) & = \frac{\rho_u D_u}{\rho_u R_u^S+(1-\zeta_u)F_u}, \\
	\nonumber V_u^{\eta-\mathrm{opt}}(B_u,R_u^S,F_u) & = \frac{D_u}{B_u r_u}\left[B_u r_u - R_u^S \right. \\
	& \left. \hspace{20mm} - (1-\zeta_u) \frac{F_u}{\rho_u}\right].
\end{align}
This verifies the fourth row of Table \ref{eta-opt}.

We then assume
\begin{equation}
	\frac{R_u^S}{\zeta_u} \leq \frac{F_u}{\rho_u} < B_u r_u.
\end{equation}
In this case, when the following inequality holds
\begin{equation}
	0 \leq \eta_u < \frac{F_u}{\rho_u B_u r_u},
\end{equation}
the expressions of $T_u$ and $V_u$ refer to (\ref{expr_T}d) and (\ref{expr_V}d), where both $T_u$ and $V_u$ monotonically decrease with $\eta_u$ increasing. When the following inequality holds
\begin{equation}
	\frac{F_u}{\rho_u B_u r_u} \leq \eta_u \leq 1,
\end{equation}
the expressions of $T_u$ and $V_u$ refer to (\ref{expr_T}a) and (\ref{expr_V}a). With $\eta_u$ increasing, $T_u$ monotonically decrease while $V_u$ remain unchanged. Therefore, on the condition of
\begin{equation}
	B_u r_u \geq \frac{R_u^S}{\zeta_u},\ \frac{R_u^S}{\zeta_u} \leq \frac{F_u}{\rho_u} < B_u r_u,
\end{equation}
the optimal value of $\eta_u$ to minimize both functions is determined to be
\begin{equation}
	\eta_u^{\mathrm{opt}} = 1,
\end{equation}
and the corresponding function values are
\begin{align}
	T_u^{\eta-\mathrm{opt}}(B_u,R_u^S,F_u) & = \frac{\zeta_u D_u}{R_u^S}, \\
	\nonumber V_u^{\eta-\mathrm{opt}}(B_u,R_u^S,F_u) & = \frac{D_u}{B_u r_u}\left[B_u r_u - R_u^S \right. \\
	& \hspace{20mm} \left. - (1-\zeta_u) \frac{F_u}{\rho_u}\right].
\end{align}
This verifies the fifth row of Table \ref{eta-opt}.

Finally we assume
\begin{equation}
	\frac{F_u}{\rho_u} \geq B_u r_u.
\end{equation}
In this case, for
\begin{equation}
	0 \leq \eta_u \leq 1
\end{equation}
the expressions of $T_u$ and $V_u$ refer to (\ref{expr_T}d) and (\ref{expr_V}d), where both $T_u$ and $V_u$ monotonically decrease with $\eta_u$ increasing. Therefore, on the condition of
\begin{equation}
	B_u r_u \geq \frac{R_u^S}{\zeta_u},\ \frac{F_u}{\rho_u} \geq B_u r_u,
\end{equation}
the optimal value of $\eta_u$ to minimize both functions is determined to be
\begin{equation}
	\eta_u^{\mathrm{opt}} = 1,
\end{equation}
and the corresponding function values are
\begin{align}
	T_u^{\eta-\mathrm{opt}}(B_u,R_u^S,F_u) & = \frac{\zeta_u D_u}{R_u^S}, \\
	V_u^{\eta-\mathrm{opt}}(B_u,R_u^S,F_u) & = \frac{D_u}{B_u r_u}(\zeta_u B_u r_u - R_u^S).
\end{align}
This verifies the sixth row of Table \ref{eta-opt}.

\section{Proof of \textit{Theorem 3}}
We refer to the proof of \textit{Theorem 2} and also adopt the method of proof by contradiction. Assume that the theorem does not hold, which is equivalent to the following statement. For problem (\ref{problem 2}) there exists at least one set of optimal variables $(\mathbf{B}_{0}, \mathbf{R}^S_{0}, \mathbf{F}_{0})$ lying outside the feasible region of problem (\ref{problem 3}), and satisfying that
\begin{equation}
	\underset{u}{\max} \  T_u^{\eta-\mathrm{opt}}(B_{0,u},R_{0,u}^S,F_{0,u}) < \underset{u}{\max} \  \frac{D_u}{B_u r_u}
\end{equation}
holds for all feasible solutions $(\mathbf{B}, \mathbf{R}^S, \mathbf{F})$ of problem (\ref{problem 3}). We categorize the elements of $(\mathbf{B}_{0}, \mathbf{R}^S_{0}, \mathbf{F}_{0})$ based on the relationship among $B_{0,u}$, $R^S_{0,u}$ and $F_{0,u}$.

Select all elements from $(\mathbf{B}_{0}, \mathbf{R}^S_{0}, \mathbf{F}_{0})$ that satisfy
\begin{equation}
	B_{0,u} r_u < R_{0,u}^S,
\end{equation}
corresponding to the first row of Table \ref{eta-opt}. For any one of those elements $(B_{0,u}, R^S_{0,u}, F_{0,u})$, we present a corresponding solution $(B_{1,u}, R^S_{1,u}, F_{1,u})$ which satisfies
\begin{subnumcases}{}
	B_{1,u} = B_{0,u} \\
	R^S_{1,u} = B_{0,u} r_u \\
	F_{1,u} = 0
\end{subnumcases}
We can first verify that
\begin{equation}
	(B_{1,u}, R^S_{1,u}, F_{1,u}) \leq (B_{0,u}, R^S_{0,u}, F_{0,u})
\end{equation}
holds. Besides, we can verify that
\begin{equation}
	R_{1,u}^S \leq B_{1,u} r_u \leq \frac{R_{1,u}^S}{\zeta_u}
\end{equation}
and
\begin{equation}
	\frac{F_{1,u}}{\rho_u} = \frac{B_{1,u} r_u - R_{1,u}^S}{1-\zeta_u}
\end{equation}
also hold. Finally, we can verify that
\begin{equation}
	T_u^{\eta-\mathrm{opt}}(B_{0,u},R_{0,u}^S,F_{0,u}) = \frac{D_u}{B_{0,u} r_u} = \frac{D_u}{B_{1,u} r_u}.
\end{equation}

Then select all elements from $(\mathbf{B}_{0}, \mathbf{R}^S_{0}, \mathbf{F}_{0})$ that satisfy
\begin{align}
	\nonumber & R_{0,u}^S \leq B_{0,u} r_u < \frac{R_{0,u}^S}{\zeta_u}, \\
	& \hspace{30mm} \frac{F_{0,u}}{\rho_u} < \frac{B_{0,u} r_u - R_{0,u}^S}{1-\zeta_u},
\end{align}
or
\begin{equation}
	B_{0,u} r_u \geq \frac{R_{0,u}^S}{\zeta_u},\  \frac{F_{0,u}}{\rho_u} < \frac{R_{0,u}^S}{\zeta_u},
\end{equation}
corresponding to the second and fourth row of Table \ref{eta-opt}. For any one of those elements $(B_{0,u}, R^S_{0,u}, F_{0,u})$, we present a corresponding solution $(B_{1,u}, R^S_{1,u}, F_{1,u})$ which satisfies
\begin{subnumcases}{}
	B_{1,u} = \frac{1}{r_u} \left(R_{0,u}^S + \frac{(1-\zeta_u)F_{0,u}}{\rho_u} \right) \\
	R^S_{1,u} = R^S_{0,u} \\
	F_{1,u} = F_{0,u}
\end{subnumcases}
We can first verify that
\begin{equation}
	(B_{1,u}, R^S_{1,u}, F_{1,u}) \leq (B_{0,u}, R^S_{0,u}, F_{0,u})
\end{equation}
holds. Besides, we can verify that
\begin{equation}
	R_{1,u}^S \leq B_{1,u} r_u \leq \frac{R_{1,u}^S}{\zeta_u}
\end{equation}
and
\begin{equation}
	\frac{F_{1,u}}{\rho_u} = \frac{B_{1,u} r_u - R_{1,u}^S}{1-\zeta_u}
\end{equation}
also hold. Finally, we can verify that
\begin{align}
	\nonumber T_u^{\eta-\mathrm{opt}} & (B_{0,u},R_{0,u}^S,F_{0,u}) = \frac{D_u \rho_u}{F_{0,u}(1-\zeta_u)+\rho_u R_{0,u}^S} \\
	& = \frac{D_u \rho_u}{F_{1,u}(1-\zeta_u)+\rho_u R_{1,u}^S} = \frac{D_u}{B_{1,u} r_u}.
\end{align}

Next select all elements from $(\mathbf{B}_{0}, \mathbf{R}^S_{0}, \mathbf{F}_{0})$ that satisfy
\begin{align}
	\nonumber & R_{0,u}^S \leq B_{0,u} r_u < \frac{R_{0,u}^S}{\zeta_u}, \\
	& \hspace{30mm} \frac{F_{0,u}}{\rho_u} > \frac{B_{0,u} r_u - R_{0,u}^S}{1-\zeta_u},
\end{align}
corresponding to the third row of Table \ref{eta-opt}. For any one of those elements $(B_{0,u}, R^S_{0,u}, F_{0,u})$, we present a corresponding solution $(B_{1,u}, R^S_{1,u}, F_{1,u})$ which satisfies
\begin{subnumcases}{}
	B_{1,u} = B_{0,u} \\
	R^S_{1,u} = R^S_{0,u} \\
	F_{1,u} = \rho_u \frac{B_{0,u} r_u - R^S_{0,u}}{1-\zeta_u}
\end{subnumcases}
We can first verify that
\begin{equation}
	(B_{1,u}, R^S_{1,u}, F_{1,u}) \leq (B_{0,u}, R^S_{0,u}, F_{0,u})
\end{equation}
holds. Besides, we can verify that
\begin{equation}
	R_{1,u}^S \leq B_{1,u} r_u \leq \frac{R_{1,u}^S}{\zeta_u}
\end{equation}
and
\begin{equation}
	\frac{F_{1,u}}{\rho_u} = \frac{B_{1,u} r_u - R_{1,u}^S}{1-\zeta_u}
\end{equation}
also hold. Finally, we can verify that
\begin{equation}
	T_u^{\eta-\mathrm{opt}}(B_{0,u},R_{0,u}^S,F_{0,u}) = \frac{D_u}{B_{0,u} r_u} = \frac{D_u}{B_{1,u} r_u}.
\end{equation}

Finally select all elements from $(\mathbf{B}_{0}, \mathbf{R}^S_{0}, \mathbf{F}_{0})$ that satisfy
\begin{equation}
	B_{0,u} r_u \geq \frac{R_{0,u}^S}{\zeta_u},\  \frac{F_{0,u}}{\rho_u} \geq \frac{R_{0,u}^S}{\zeta_u},
\end{equation}
corresponding to the fifth and sixth row of Table \ref{eta-opt}. For any one of those elements $(B_{0,u}, R^S_{0,u}, F_{0,u})$, we present a corresponding solution $(B_{1,u}, R^S_{1,u}, F_{1,u})$ which satisfies
\begin{subnumcases}{}
	B_{1,u} = \frac{R^S_{0,u}}{\zeta_u r_u}\\
	R^S_{1,u} = R^S_{0,u} \\
	F_{1,u} = \rho_u \frac{R^S_{0,u}}{\zeta_u}
\end{subnumcases}
We can first verify that
\begin{equation}
	(B_{1,u}, R^S_{1,u}, F_{1,u}) \leq (B_{0,u}, R^S_{0,u}, F_{0,u})
\end{equation}
holds. Besides, we can verify that
\begin{equation}
	R_{1,u}^S \leq B_{1,u} r_u \leq \frac{R_{1,u}^S}{\zeta_u}
\end{equation}
and
\begin{equation}
	\frac{F_{1,u}}{\rho_u} = \frac{B_{1,u} r_u - R_{1,u}^S}{1-\zeta_u}
\end{equation}
also hold. Finally, we can verify that
\begin{align}
	\nonumber T_u^{\eta-\mathrm{opt}}(B_{0,u},R_{0,u}^S, & F_{0,u}) = \frac{\zeta_u D_u}{R_{0,u}^S} \\
	& = \frac{\zeta_u D_u}{R_{1,u}^S} = \frac{D_u}{B_{1,u} r_u}.
\end{align}

Therefore, we could obtain a solution $(\mathbf{B}_{1}, \mathbf{R}^S_{1}, \mathbf{F}_{1})$. From the above analysis we know that
\begin{equation}
	(\mathbf{B}_{1}, \mathbf{R}^S_{1}, \mathbf{F}_{1}) \leq (\mathbf{B}_{0}, \mathbf{R}^S_{0}, \mathbf{F}_{0}),
\end{equation} 
which means that the solution $(\mathbf{B}_{1}, \mathbf{R}^S_{1}, \mathbf{F}_{1})$ satisfy constraints (\ref{problem 3b})-(\ref{problem 3d}) since $(\mathbf{B}_{0}, \mathbf{R}^S_{0}, \mathbf{F}_{0})$ satisfy (\ref{problem 2c})-(\ref{problem 2e}). Besides, from the above analysis, the solution $(\mathbf{B}_{1}, \mathbf{R}^S_{1}, \mathbf{F}_{1})$ also satisfy constraints (\ref{problem 3e}) and (\ref{problem 3f}), and it could be easily verified that $(\mathbf{B}_{1}, \mathbf{R}^S_{1}, \mathbf{F}_{1})$ also satisfy (\ref{problem 3g})-(\ref{problem 3i}). Therefore, $(\mathbf{B}_{1}, \mathbf{R}^S_{1}, \mathbf{F}_{1})$ is a feasible solution of problem (\ref{problem 3}).

From the above analysis, we also know that
\begin{equation}
	\underset{u}{\max} \  T_u^{\eta-\mathrm{opt}}(B_{0,u},R_{0,u}^S,F_{0,u}) = \underset{u}{\max} \  \frac{D_u}{B_{1,u} r_u}
\end{equation}
which causes contradiction to the original assumption. Thus, the assumption does not hold and the theorem is proved.

\section{Proof of \textit{Theorem 5}}
To prove the theorem, we first consider a new optimization problem
\begin{subequations} \label{problem 5}
	\begin{align}
		\min_{\mathbf{B}, \mathbf{R}^S, T} \  & T \label{problem 5a} \\
		\mathrm{s.t.} \ \ & T \geq \frac{D_u}{B_u r_u}, \  \forall u = 1,...,U, \label{problem 5b} \\
		& \sum_{u=1}^{U} B_u \leq B_{\mathrm{total}}, \label{problem 5c} \\
		& \sum_{u=1}^{U} R_u^S \leq R^S_{\mathrm{total}}, \label{problem 5d} \\
		& R_u^S \leq B_u r_u \leq \frac{R_u^S}{\zeta_u}, \  \forall u = 1,...,U, \label{problem 5e} \\
		& T \geq 0, \label{problem 5h} \\
		& B_u \geq 0, \  \forall u = 1,...,U, \label{problem 5f} \\
		& R_u^S \geq 0, \  \forall u = 1,...,U. \label{problem 5g}
	\end{align}
\end{subequations}
We can see that problem (\ref{problem 5}) is obtained by removing the constraint (\ref{problem 4e}) from problem (\ref{problem 4}). For this problem, there is at least one set of optimal variables $(\mathbf{B}_0, \mathbf{R}^S_0, T_0)$ that satisfies
\begin{equation}
	T_0 = \frac{D_u}{B_{0,u} r_u}, \forall u = 1,...,U. \label{appd4.1}
\end{equation}
The reason is that if there exists a set of optimal variables $(\mathbf{B}_1, \mathbf{R}^S_1, T_1)$ that does not satisfy this condition, we could generate a new set of variable values as follows
\begin{subnumcases}{}
	B_{0,u} = \frac{D_u}{T_{1} r_u}, \forall u = 1,...,U, \\
	R^S_{0,u} = \frac{R^S_{1,u}}{B_{1,u}} \cdot \frac{D_u}{T_{1} r_u}, \forall u = 1,...,U, \\
	T_{0} = T_{1}.
\end{subnumcases}
It can be easily verified that this solution satisfies (\ref{appd4.1}). Since $T_{0} = T_{1}$, this set of variable values is also optimal. Therefore, the following optimization problem has the same optimal objective function value as problem (\ref{problem 5}).
\begin{subequations} \label{problem 6}
	\begin{align}
		\min_{\mathbf{R}^S, T} \  & T \label{problem 6a} \\
		\mathrm{s.t.} \ \ & \sum_{u=1}^{U} \frac{D_u}{r_u T} \leq B_{\mathrm{total}}, \label{problem 6b} \\
		& \sum_{u=1}^{U} R_u^S \leq R^S_{\mathrm{total}}, \label{problem 6c} \\
		& \zeta_u \frac{D_u}{T} \leq R_u^S \leq \frac{D_u}{T} , \  \forall u = 1,...,U, \label{problem 6d} \\
		& T \geq 0, \label{problem 6e} \\
		& R_u^S \geq 0, \  \forall u = 1,...,U. \label{problem 6f}
	\end{align}
\end{subequations}
We could introduce a set of medium variables $\bm{\omega} = [\omega_1,...,\omega_u,...,\omega_U]$ that satisfies
\begin{equation}
	R_u^S = \omega_u \frac{D_u}{T},
\end{equation}
and problem (\ref{problem 6}) can be equivalently transformed into the following optimization problem
\begin{subequations} \label{problem 7}
	\begin{align}
		\min_{\bm{\omega}, T} \  & T \label{problem 7a} \\
		\mathrm{s.t.} \ \ & T \geq \frac{1}{B_{\mathrm{total}}} \sum_{u=1}^{U} \frac{D_u}{r_u}, \label{problem 7b} \\
		& T \geq \frac{1}{R^S_{\mathrm{total}}} \sum_{u=1}^{U} \omega_u D_u, \label{problem 7c} \\
		& T \geq 0, \label{problem 7d} \\
		& \zeta_u \leq \omega_u \leq 1, \  \forall u = 1,...,U. \label{problem 7e}
	\end{align}
\end{subequations}
We could easily obtain the optimal solution to this problem, which could be given by
\begin{equation}
	\omega_u = \zeta_u, \  \forall u = 1,...,U,
\end{equation}
\begin{equation}
	T = \max \left\{ \frac{\sum_{u=1}^{U} \frac{D_u}{r_u}}{B_{\mathrm{total}}} , \frac{\sum_{u=1}^{U} \omega_u D_u}{R^S_{\mathrm{total}}}\right\}.
\end{equation}
Therefore, by substituting the medium variables $\bm{\omega}$, we can obtain a set of optimal variables to problem (\ref{problem 5}) written as follows
\begin{align}
	\nonumber \overline{B_{u}} & = \frac{D_u}{r_u} \min \left\{ \frac{B_{\mathrm{total}}}{\sum_{u=1}^{U} \frac{D_u}{r_u}} , \frac{R^S_{\mathrm{total}}}{\sum_{u=1}^{U} \zeta_u D_u} \right\}, \\
	& \hspace{47mm} \forall u = 1,...,U, \\
	\nonumber \overline{R^S_{u}} & = \zeta_u D_u \min \left\{ \frac{B_{\mathrm{total}}}{\sum_{u=1}^{U} \frac{D_u}{r_u}} , \frac{R^S_{\mathrm{total}}}{\sum_{u=1}^{U} \zeta_u D_u} \right\}, \\
	& \hspace{47mm} \forall u = 1,...,U, \\
	\overline{T} & = \max \left\{ \frac{\sum_{u=1}^{U} \frac{D_u}{r_u}}{B_{\mathrm{total}}} , \frac{\sum_{u=1}^{U} \zeta_u D_u}{R^S_{\mathrm{total}}} \right\}.
\end{align}
Moreover, since that the inequality (\ref{F cond}) holds, we can verify that the above solution also ensures constraint (\ref{problem 4e}) always holds. Since problem (\ref{problem 5}) is obtained by removing the constraint (\ref{problem 4e}) from problem (\ref{problem 4}), the above solution is also an optimal solution for problem (\ref{problem 4}), which proves the theorem.

\end{document}